\documentclass{amsart}

\usepackage{graphicx}
\usepackage{amsmath,amssymb,latexsym} 
\usepackage{etoolbox}
\usepackage{enumitem}
\usepackage{color}
\usepackage{tikz}
\usetikzlibrary{arrows,shapes, calc, fit, positioning}
\usepackage{caption}
\usepackage{subcaption}
\usepackage[ruled,vlined,linesnumbered]{algorithm2e}

\SetCommentSty{mycommfont}
\usepackage{mathtools}
\usepackage{todonotes}
\usepackage[left=1.2in,right=1.2in,top=1in,bottom=1in]{geometry}

\theoremstyle{plain}
\newtheorem{theorem}{Theorem}
\newtheorem{lemma}{Lemma}
\newtheorem{proposition}{Proposition}

\theoremstyle{definition}

\usepackage{tikz,enumitem,framed}

\def\le{\leqslant}
\def\ge{\geqslant}
\def\T{\mathsf{T}}
\def\Z{\mathbb{Z}}
\def\R{\mathbb{R}}

\def\a{\boldsymbol{a}}
\def\cc{\boldsymbol{c}}
\def\e{\boldsymbol{e}}
\def\x{\boldsymbol{x}}
\def\y{\boldsymbol{y}}
\def\z{\boldsymbol{z}}
\def\w{\boldsymbol{w}}
\def\vv{\boldsymbol{v}}

\def\calA{\mathcal{A}}
\def\calL{\mathcal{L}}

\def\med{\mathrm{med}~}

\def\bigL{\mathsf{L}}
\def\K{\mathsf{K}}
\def\id{\operatorname{id}}

\newcommand{\frederic}[1]{\textcolor{blue}{#1}}

\begin{document}
\title[]{Envy-free division of multi-layered cakes}
\author{Ayumi Igarashi and Fr\'ed\'eric Meunier}

\maketitle

\begin{abstract}
We study the problem of dividing a multi-layered cake under non-overlapping constraints. 
This problem, recently proposed by Hosseini et al. (IJCAI, 2020), captures several natural scenarios such as the allocation of multiple facilities over time where each agent can utilize at most one facility simultaneously, and the allocation of tasks over time where each agent can perform at most one task simultaneously. We establish the existence of an envy-free multi-division that is both non-overlapping and contiguous within each layer when the number of agents is a prime power and the number of layers is at most the number of agents, providing a positive partial answer to an open question by Hosseini et al. To achieve this, we employ a new approach based on a general fixed point theorem, originally proven by Volovikov (\emph{Mathematical Notes}, 1996), and recently applied by Joji\'{c}, Panina, and {\v{Z}}ivaljevi\'{c} (\emph{SIAM Journal on Discrete Mathematics}, 2020) to the envy-free division problem of a cake. We further design a fully polynomial-time approximation scheme (FPTAS) to find an approximate envy-free solution that is both non-overlapping and contiguous for a two-layered cake division among three agents with monotone preferences. 

More generally, we establish all the results for divisions among groups of almost the same size.  
When there are three groups and one layer, our FPTAS is actually able to deal with groups of any predetermined sizes, still under the monotone preference assumptions. For three groups, this provides thus an algorithmic version of a recent theorem by Segal-Halevi and Suksompong (\emph{The American Mathematical Monthly}, 2021). A classical lemma from topology ensures that every self-map of a simplex mapping each facet to itself is surjective. This lemma has often been used in economics theory and our algorithm can be interpreted as an FPTAS counterpart of its two-dimensional version.
\end{abstract}

\section{Introduction}
Imagine $n$ students taking a gymnasium course. The professor of the course wishes to divide the students into groups and assign $m$ activities over a given period of time. Students may have different opinions about the time slot and activity to which they would like to be assigned; for instance, some may prefer to swim in the morning and play basketball in the afternoon, while the others may want to play ping pong as long as possible. 

The problem of fairly distributing a resource has been often studied in the classical cake-cutting model \cite{Steinhaus48}, where the cake, represented by the unit interval, is to be divided among heterogenous agents. A variety of cake-cutting techniques have been developed over the past decades; in particular, the existence and algorithmic questions concerning an \emph{envy-free} division, where each agent receives her first choice under the given division, 
have turned out to involve highly non-trivial arguments \cite{DubinsSpanier,stromquist1980cut,su1999rental,woodall1980dividing}. 

In the above example of assigning multiple activities, however, one cannot trivially reduce the problem to the one-dimensional case. Indeed, if we merely divide the $m$ time intervals independently, this may result in a division that is not feasible, i.e., it may assign overlapping time intervals to the same agent who can perform at most one activity at a given time. 
In order to capture such constraints, Hosseini, Igarashi, and Searns \cite{hosseini2020fair} have recently introduced the multi-layered cake-cutting problem. There, $n$ agents divide a cake formed by $m$ different layers under the \emph{feasibility} constraint: the pieces of different layers assigned to the same agent should be non-overlapping, i.e., these pieces should have disjoint interiors. Besides an application to assign activities, the model can capture a plethora of real-life situations. For instance, 
consider a situation with several tasks running all day long in parallel (e.g., a fast-food), and workers turning between these tasks and who cannot be assigned to more than one task simultaneously.

For the special case of two layers and two agents, Hosseini et al. showed that the cut-and-choose protocol achieves envy-free division that is both feasible and contiguous within each layer by a single cut located at the same position over the two layers; in other words, the division is obtained by a ``long knife."
For a more general combination of positive integers $m$ and $n$ with $m \le n$, it has remained an open question whether there is an envy-free multi-division that is both contiguous and feasible, even in the special case when $m=2$ and $n=3$. 
Though, when the contiguity requirement is relaxed, they showed the existence of envy-free feasible multi-divisions~\cite{hosseini2020arxiv}. 
Note that when the number of layers strictly exceeds the number of agents, i.e., $m>n$, there is no way to allocate the entire layered cake while satisfying feasibility.

Now, returning to our first example, recall that in that example, we wish to assign activities to \emph{groups} of students, instead of \emph{individuals}. 
In the above example with workers assigned to tasks, it also makes sense to consider a variant with groups of workers staying together all day long. In this paper, we consider a ``group'' generalization of the multi-layered cake-cutting problem introduced by Hosseini et al.~\cite{hosseini2020fair}. The aim is to divide a multi-layered cake into $q$ pieces, partition $n$ agents into $q$ groups of almost equal size, and assign the pieces to groups in a fair manner.
Our focus is on envy-free divisions of a multi-layered cake under feasible and contiguous constraints. 

We first show that when $q$ is a prime power and $m \le q\le n$, an envy-free multi-division that is both feasible and contiguous exists under mild conditions on preferences (especially they are not required to be monotone). Our proof relies on a general Borsuk--Ulam-type theorem proven by Volovikov~\cite{volovikov1996topological}, whose application for this type of problems has been initiated by Joji{\'c}, Panina, and {\v{Z}}ivaljevi\'{c} \cite{jojic2019splitting}; we comment on this proof technique in the beginning of Section~\ref{sec:long}. Further, such divisions can be obtained using $q-1$ cuts in each layer located at the same positions, i.e., the divisions are obtained by $q-1$ long knives. 
(In the initial example of a gymnasium, we can ensure for free that all activities begin and end at the same time.) Note that when $q=n$, our problem reduces to the setting of Hosseini et al., which means that the present paper settles their open question partially. 
Unfortunately, our existence result is the best one could hope to obtain, under the general preference model with choice functions: 
Avvakumov and Karasev \cite{Avvakumov_2020} and Panina and {\v{Z}}ivaljevi\'c \cite{panina2021envyfree} provided examples of a cake-cutting instance with choice functions for which no envy-free division exists, for every choice of $q$ that is not a prime power.

Our existence result concerning envy-freeness also answers another open question raised by Hosseini et al.: 
By a recursive procedure initialized with the above existence result, we prove the existence of a {\em proportional} multi-division that is feasible and contiguous for any $q$ with $m \le q  \le n$ when agents have additive valuations, and thus properly generalizes the known existence result when $m$ is a power of $2$, $m \le n$, and $q=n$~\cite{hosseini2020fair}; see Section \ref{sec:concluding} for details.

Hosseini et al. also discussed the question of a procedure achieving an envy-free multi-division in the two-layered case. They considered divisions obtained by cutting the top layer with one ``short knife'' and dividing the rest with one ``long knife.'' For this particular problem, we observe that one can encode such divisions by the points of the unit square. This way, a Sperner-type argument turns out to be applicable: there is an envy-free multi-division using only one short knife and one long knife when agents have monotone preferences. By exploiting the monotonicity that arises when fixing the long knife position, we further devise a fully polynomial-time approximation scheme (FPTAS) to compute an approximate envy-free multi-division. Moreover, when there is only one layer, the algorithm can handle the case where the sizes of the groups can be set arbitrarily, which corresponds to the algorithmic version of the three-group case of the recent existence result by Segal-Halevi and Suksompong~\cite{Segal2021}. 

Most results of the paper are actually stated and proved for \emph{birthday cake multi-divisions}, i.e., divisions of the cake such that whichever piece a \emph{birthday agent} selects, there is an envy-free assignment of the remaining pieces to the remaining agents.

\smallskip
\noindent
{\bf Related work}
Early literature on cake-cutting has established the existence of an envy-free contiguous division under \emph{closed} (if the $i$-th piece is preferred in a convergent sequence of divisions, it is preferred in the limit) and \emph{hungry} preferences (pieces of nonzero-length are preferred over pieces of zero-length). While classical works~\cite{stromquist1980cut,woodall1980dividing} applied non-constructive topological proofs, Su~\cite{su1999rental} provided a more combinatorial argument by explicitly using Sperner's lemma. 
Recently, the problem of dividing a partially unappetizing cake has attracted a great deal of attention. Here, some agents may find that a part of the cake is unappetizing and prefer nothing, while others may find it tasty (agents may have \emph{non-hungry} preferences). 
Even in such cases, an envy-free division only using $n-1$ cuts has been shown to exist for a particular number $n$ of agents under the assumption of closed preferences \cite{Avvakumov_2020,jojic2019splitting,meunier2019envy, panina2021envyfree,Halevi2018}. The most general result obtained so far is the one by Avvakumov and Karasev \cite{Avvakumov_2020}, who showed the existence of an envy-free division for the case when $n$ is a prime power.
Joji{\'c}, Panina, and {\v{Z}}ivaljevi\'{c} \cite{jojic2019splitting} gave an alternative proof of the result of Avvakumov and Karasev, 
by using Volovikov's theorem \cite{volovikov1996topological}.
Pania and {\v{Z}}ivaljevi\'{c}~\cite{panina2021envyfree} pushed further this technique to refine this kind of results allowing non-hungry preferences. 

An important tool in the proof of Avvakumov and Karasev, as well as in the proof of Joji{\'c}, Panina, and {\v{Z}}ivaljevi\'{c}, is {\em Gale's averaging trick}, introduced by Gale~\cite{gale1984equilibrium}, in the context of an exchange economy. As far as we know, the first paper applying Gale's averaging trick in the context of envy-free cake-cutting is a paper by Asada et al.~\cite{Asada2018}. Roughly speaking, it consists in considering an aggregated preference function to which a topological result is applied, and then in recovering information for each agent by applying some flow argument. In the present work, we use Volovikov's theorem and the Gale's averaging trick.

In his classical work on cake-cutting, Woodall~\cite{woodall1980dividing} also proved that an envy-free division can be obtained without knowing one agent's preference. 
For example, the cut-and-choose protocol does not need the chooser's preference to obtain an envy-free division. More generally, for any number $n$ of agents, there is a division of the cake into $n$ contiguous pieces such that whichever piece a birthday agent selects, there is an envy-free assignment of the remaining pieces to the remaining agents. Asada et al.~\cite{Asada2018} gave a simple combinatorial proof that shows the existence of such division. 

The vast majority of the literature on cake-cutting is concerned with an allocation among single agents. A notable exception is the work of Segal-Halevi and Suksompong~\cite{SSgroup}, who introduced the standard cake-cutting problem among groups of agents. They established the existence of an envy-free division among groups of varying sizes, showing that it is possible to partition a cake into $q$ contiguous pieces as well as agents into $q$ groups of any desired sizes, and assign the pieces to the groups so that no agent prefers a piece assigned to another group to the piece assigned to her own group.

In general, there is no finite protocol that computes an exact envy-free division even for three agents \cite{Stro08a}, though such protocol exists when relaxing the contiguity requirement \cite{AzizM16}. Nevertheless, for three agents with monotone valuations, Deng et al.~\cite{Deng2012} proved that an $\varepsilon$-approximate envy-free division can be computed in logarithmic time of $\frac{1}{\varepsilon}$, while obtaining PPAD-hardness of the same problem for choice functions whose choice is given explicitly by polynomial time algorithms. Our method for establishing the FPTAS exploits monotonicity in agents' preferences and employs a divide-and-conquer approach similar to that of Deng et al. A difference is that while Deng et al. used a triangulation of the two-dimensional standard simplex, we subdivide the unit square into small squares and compute by binary-search thinner and thinner full-height rectangles containing an approximate solution. Our algorithmic result holds for a more general version, considering simultaneously a birthday agent, groups, and two layers. 

In a different context, Deligkas et al.~\cite{Deligkas2021} also observed that partial monotonicity is useful to design a similar binary-search algorithm to find a consensus-halving among two agents. There, one agent has a monotone preference and another agent has a continuous preference. We would like to emphasize, however, that our result for a two-layered cake does not assume continuity of the birthday agent while they assume continuity.

Several papers also studied the fair division problem in which agents divide multiple cakes~\cite{cloutier2010two,lebert2013envy,nyman2020fair,Segal2021}. This model requires each agent to receive at least one nonempty piece of each cake \cite{cloutier2010two,lebert2013envy,nyman2020fair} or receive pieces on as few cakes as possible~\cite{Segal2021}. On the other hand, our setting requires the allocated pieces to be non-overlapping. Thus, the existence/non-existence of envy-free divisions in one setting does not imply those for another.

\section{Model}\label{sec:prelim}
We consider the setting of Hosseini, Igarashi, and Searns \cite{hosseini2020fair}, except that we allow slightly more general preferences and that we aim to obtain a division among groups of agents. We are given $m$ layers, $n$ agents, and a positive integer $q$ with $1 \le q \le n$. A \emph{cake} is the unit interval $[0,1]$. A {\em piece} of cake is a union of finitely many disjoint closed subintervals of $[0,1]$. We refer to a subinterval of $[0,1]$ as a {\em contiguous piece} of cake. An {\em $m$-layered cake} is a sequence of $m$ cakes $[0,1]$, each being a {\em layer}. A {\em layered piece} is a sequence $\calL=(L_{\ell})_{\ell \in [m]}$ of pieces of each layer $\ell$; a layered piece is {\em contiguous} if each $L_{\ell}$ is a contiguous piece of layer $\ell \in [m]$. 
A layered piece $\calL$ is {\em non-overlapping} if no two pieces from different layers overlap, i.e., $L_{\ell} \cap L'_{\ell}$ is empty or formed by finitely many points for every pair of distinct layers $\ell,\ell'$. 
The {\em length} of a layered piece is the sum of the lengths of its pieces in each layer. 
A {\em multi-division} $\calA=(\calA_1,\calA_2,\ldots,\calA_q)$ is a $q$-tuple forming a partition of the $m$-layered cake into $q$ layered pieces.
(Here, ``partition'' is used in a slightly abusive way: while the collection covers the layered cake and the layered pieces have disjoint interiors, we allow the latter to share endpoints.) A multi-division $\calA$ is 
\begin{itemize}
\item {\em contiguous} if $\calA_i$ is contiguous for each $i \in [q]$. 
\item {\em feasible} if $\calA_i$ is non-overlapping for each $i \in [q]$. 
\end{itemize}
We focus in this work on \emph{complete} multi-divisions where the entire layered cake must be allocated.

Each agent $i$ has a {\em choice function} $c_i$ that, given a multi-division, returns the set of {\em preferred} layered pieces (among which the agent is indifferent). This function returns the same set of pieces over all permutations of the entries of the multi-division:
\[
c_i(\calA_1,\calA_2,\ldots,\calA_q) = c_i(\calA_{\rho(1)},\calA_{\rho(2)},\ldots,\calA_{\rho(q)}) \quad \forall \rho \in \mathcal{S}_q \, .
\]
The choice function model is used in \cite{MeunierSu,su1999rental} and more general than the valuation model while the latter is more standard in fair division.

An agent $i$ \emph{weakly prefers} a layered piece $\calA_j$ to another layered piece $\calA_{j'}$ within a multi-division $\calA$ if
\[
\calA_{j'} \in c_i(\calA) \qquad \Longrightarrow \qquad \calA_{j} \in c_i(\calA) \, .
\]
An agent has {\em hungry preferences} if in any multi-division every layered piece of nonzero-length is weakly preferred to every layered piece of zero-length. An agent $i$ has {\em monotone preferences} if every pair $\calA, \calA'$ of multi-divisions with a $j$ such that $\calA_j \subseteq \calA'_j$ and $\calA_{j'} \supseteq \calA'_{j'}$ for all $j' \neq j$ satisfies
\[
\calA_j \in c_i(\calA) \qquad \Longrightarrow \qquad \calA'_j \in c_i(\calA') \, .
\]
An agent $i$ has {\em closed preferences} if the following holds: for every sequence $(\calA^{(t)})_{t\in\Z_+}$ of multi-divisions converging to a multi-division $\calA^{(\infty)}$, we have
\[
\calA_j^{(t)} \in c_i(\calA^{(t)})\quad \forall t\in\Z_+ \qquad \Longrightarrow \qquad \calA_j^{(\infty)} \in c_i(\calA^{(\infty)}) \, .
\]
The convergence of layered pieces is considered according to the pseudo-metric $d(\calL,\calL')=\mu(\calL \triangle \calL')$; a sequence of multi-divisions is converging if each of its layered pieces converges. Here, $\mu$ is the Lebesgue measure and $\calL \triangle \calL'=( (L_{\ell} \setminus L'_{\ell}) \cup (L'_{\ell} \setminus L_{\ell}))_{\ell \in [m]}$.
A multi-division $\calA=(\calA_1,\calA_2,\ldots,\calA_q)$ is {\em envy-free} if there exists a surjective assignment $\pi\colon[n] \rightarrow [q]$ such that $\calA_{\pi(i)} \in c_i(\calA)$ for all $i \in [n]$. Note that when $q=n$, our definition coincides with the standard definition of an envy-free division that assigns each piece to a single agent. 
In general, our results are stated with a \emph{birthday agent}, who is not taken into account for defining or computing multi-divisions, but is still considered for the overall envy-freeness.

We also consider a setting where each agent can specify the valuation of each layered piece. Each agent $i$ has a {\em valuation function} $v_i$ that assigns a real value $v_i(\calL)$ to any layered piece $\calL$. 
A valuation function naturally gives rise to a choice function that among several layered pieces, returns the most valuable layered pieces.
A valuation function $v_i$ satisfies 
\begin{itemize}
    \item {\em monotonicity} if $v_i(\calL) \le v_i(\calL')$ for any pair of layered pieces $\calL, \calL'$ such that $L_{\ell} \subseteq L'_{\ell}$ for every $\ell \in [m]$. 
     \item {\em the Lipschitz condition} if there exists a fixed constant $K$ such that for every pair of layered pieces $\calL,\calL'$,  $|v_i(\calL)-v_i(\calL')| \le K \times \mu(\calL \triangle \calL')$.
\end{itemize}
It is easy to see that monotonicity along with the Lipschitz condition implies that the hungry assumption is satisfied: the Lipschitz condition implies that all layered pieces of zero-length have the same value; monotonicity then implies that every layered piece has a value at least that value.

For an instance with agents' valuation functions, one can define concepts of approximate envy-freeness. 
A multi-division $\calA=(\calA_1,\calA_2,\ldots,\calA_q)$ is \emph{$\varepsilon$-envy-free} if 
there exists a surjective assignment $\pi\colon[n] \rightarrow [q]$ such that for all $i \in [n]$, $v_i(\calA_{\pi(i)})+\varepsilon \ge \max_{i' \in [n]}v_i(\calA_{i'})$.
For an instance with agents' valuation functions, we assume that $v_i(\calL)$ can be accessed in constant time for any agent $i$ and layered piece $\calL$. 

There are several ways to achieve feasibility and contiguity constraints. The easiest way is probably with \emph{long knives} that cut all layers simultaneously. But we can also use \emph{short knives} that cut only a single layer at a time. See Figure~\ref{fig:long:short} for an illustration. For our result ensuring the existence of an envy-free multi-division for any number of layers, we will consider only long knives, while for the FPTAS for the two-layered cake, the multi-division will be obtained with one short and one long knife.

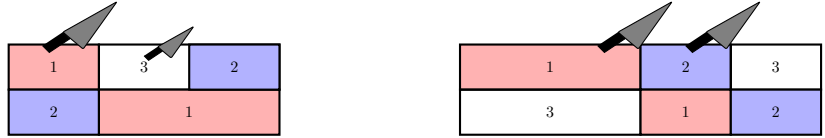
\begin{figure}[htb]
\centering
\begin{tikzpicture}[scale=0.6, transform shape]
       \draw[thick,fill=red!30] (0,0) rectangle (2,1);
        \draw[thick] (2,0) rectangle (6,1);
        \draw[thick,fill=blue!30] (0,-1) rectangle (2,0);
        \draw[thick,fill=red!30] (2,-1) rectangle (6,0);
        \draw[thick,fill=blue!30] (4,0) rectangle (6,1);
        
        \node at (1,0.5) {$1$};
        \node at (4,-0.5) {$1$};
        \node at (5,0.5) {$2$};
        \node at (1,-0.5) {$2$};
        \node at (3,0.5) {$3$};
        
        \begin{scope}[scale=1.5,xshift=0.5cm,yshift=1.4cm]
		   \draw [rounded corners=0.2mm,fill=black!50] (0.6,-0.8)--(1.1,-0.1)--(0.3,-0.55)--cycle;
		   \draw [fill=black] (0.3,-0.55) -- (0.4,-0.65) -- (0.1,-0.85)-- (0,-0.75);
		\end{scope}
		
        \begin{scope}[scale=1,xshift=3cm,yshift=1.5cm]
		   \draw [rounded corners=0.2mm,fill=black!50] (0.6,-0.8)--(1.1,-0.1)--(0.3,-0.55)--cycle;
		   \draw [fill=black] (0.3,-0.55) -- (0.4,-0.65) -- (0.1,-0.85)-- (0,-0.75);
		\end{scope}
        
        \begin{scope}[xshift=10cm,yshift=0cm]
            \draw[thick,fill=red!30] (0,0) rectangle (4,1);
            \draw[thick,fill=blue!30] (4,0) rectangle (6,1);
            \draw[thick] (6,0) rectangle (8,1);
            
            \draw[thick] (0,-1) rectangle (4,0);
            \draw[thick,fill=red!30] (4,-1) rectangle (6,0);
            \draw[thick,fill=blue!30] (6,-1) rectangle (8,0);
            
            \node at (2,0.5) {$1$};
            \node at (5,-0.5) {$1$};
            \node at (5,0.5) {$2$};
            \node at (7,-0.5) {$2$};
            \node at (7,0.5) {$3$};
            \node at (2,-0.5) {$3$};
        \end{scope}        
        
        \begin{scope}[scale=1.5,xshift=8.7cm,yshift=1.4cm]
		   \draw [rounded corners=0.2mm,fill=black!50] (0.6,-0.8)--(1.1,-0.1)--(0.3,-0.55)--cycle;
		   \draw [fill=black] (0.3,-0.55) -- (0.4,-0.65) -- (0.1,-0.85)-- (0,-0.75);
		\end{scope}
		
        \begin{scope}[scale=1.5,xshift=10cm,yshift=1.4cm]
		   \draw [rounded corners=0.2mm,fill=black!50] (0.6,-0.8)--(1.1,-0.1)--(0.3,-0.55)--cycle;
		   \draw [fill=black] (0.3,-0.55) -- (0.4,-0.65) -- (0.1,-0.85)-- (0,-0.75);
		\end{scope}
\end{tikzpicture}
\caption{Multi-divisions of a two-layered cake, obtained by one long knife and one short knife and by two long knives (pictured left-to-right).}
\label{fig:long:short}
\end{figure}

We assume basic knowledge in algebraic topology. Definitions of abstract and geometric simplicial complexes, the fact that they are somehow equivalent, and other related notions are reminded in Appendix~\ref{sec:basics}. The reader might consult the book by De Longueville~\cite{de2012course} or the one by Matou\v sek~\cite{matouvsek2003using}, especially Chapters 1 and 6 of the latter book, for complementary material. In the sequel, we will identify geometric and abstract simplicial complexes without further mention.

\section{Envy-free division using $q-1$ long knives}\label{sec:long}
Now, we formally present the first main result of this paper, stating that an envy-free multi-division using $q-1$ long knives exists when $n$ is a prime power. 

\begin{theorem}\label{thm:group-birth}
Consider an instance of the multi-layered cake-cutting problem with $m$ layers and $n$ agents, $m\le n$, with closed preferences. Let $q$ be an integer such that $m\le q\le n$. If $q$ is a prime power, then there exists a feasible and contiguous multi-division into $q$ layered pieces so that no matter which layered piece the birthday agent chooses, the remaining agents can be assigned to the layered pieces while satisfying the following two properties:
\begin{itemize}
    \item each of the remaining agents is assigned to one of her preferred layered pieces.
    \item the number of agents assigned to each layered piece, including the birthday agent, differs by at most one.
\end{itemize}
Moreover, it can be achieved with $q-1$ long knives. 
\end{theorem}

More formally (and with the birthday agent being labeled with $n$), Theorem~\ref{thm:group-birth} ensures the existence of a feasible and contiguous multi-division $\calA$ into $q$ layered pieces with the following property: for every $j^*\in[q]$, there is an assignment $\pi_{j^*}\colon[n]\rightarrow[q]$ with
\begin{itemize}
    \item $\pi_{j^*}(n)=j^*$,
    \item for each $i\in [n-1]$, $\calA_{\pi_{j^*}(i)}\in c_i(\calA)$, and
    \item for each $j\in [q]$, $|\pi^{-1}_{j^*}(j)|\in\{\lfloor n/q\rfloor, \lceil n/q\rceil \}$.
\end{itemize}

Let us comment briefly on the special case when $m=1$ and $q=n$. Since the agents might prefer zero-length pieces in our setting, our theorem boils down then to a recent result of Avvakumov and Karasev~\cite{Avvakumov_2020}. They showed that when $n$ is a prime power, there always exists an envy-free division, even if we do not assume that the agents are hungry, and that this is not true anymore if $n$ is not a prime power. (The Avvakumov--Karasev theorem was first proved for $n=3$ by Segal-Halevi~\cite{Halevi2018}---who actually initiated the study of envy-free divisions with non-necessarily hungry agents---and for prime $n$ by Meunier and Zerbib~\cite{meunier2019envy}.)

We note that the standard proof showing the existence of an envy-free division via Sperner's lemma due to Su~\cite{su1999rental} may not work in the multi-layered setting even when $q=n$. In the model of standard cake-cutting, the divisions into $q$ parallel pieces of lengths $x_i$ $(i=1,2,\ldots,q)$ can be represented by the points of the standard simplex $\Delta_{q-1}$, which is then triangulated with the vertices of each simplex being labeled with distinct owner agents, and colored in such a way that each ``owner'' agent colors the vertex with the index of her favorite bundle of the ``owned'' division. When agents always prefer nonzero-length pieces to zero-length ones, the coloring satisfies the boundary condition of Sperner's lemma. This lemma guarantees then the existence of a colorful triangle, which corresponds to an approximate envy-free division.

In the same spirit of Su's approach, one may attempt to encode feasible multi-divisions using $q-1$ long knives, by the points of the standard simplex $\Delta_{q-1}$ and apply the usual method by using Sperner's lemma to show the existence of an envy-free division. For instance, each $q$-tuple $(x_1,x_2,\ldots,x_{q})$ can represent a feasible and contiguous multi-division $(\calA_1,\ldots,\calA_q)$ where the $\ell$-th layered piece of the $i$-th bundle $\calA_i$ is given by the $\eta(i,\ell)$-piece of length $x_{\eta(i,\ell)}$
where $\eta(i,\ell)=i+\ell-1$ (modulo $m$). Unfortunately, this approach may fail to work for the multi-layered cake-cutting: even when the agents have monotone preferences over the pieces, the coloring described in the previous paragraph does not satisfy in general the boundary condition of Sperner's lemma.

We thus employ a new approach of using a general Borsuk--Ulam-type theorem, originally proven by Volovikov \cite{volovikov1996topological}, and recently applied by Joji\'{c}, Panina, and {\v{Z}}ivaljevi\'{c} \cite{jojic2019splitting,panina2021envyfree} on the envy-free division of a cake. Volovikov's theorem considers a topological space $X$ and a sphere, both on which a group of the form $((\Z_p)^k,+)$ acts (with $p$ being prime), and a map from $X$ to the sphere commuting with the action. Under some assumptions on the connectivity of $X$ and fixed-point freeness of the action on $X$, the theorem prevents the dimension of the sphere to be too small. For $q$ being a prime power $p^k$, we consider a ``configuration'' space whose points encode at the same time feasible pieces obtained by $q-1$ long knives and all possible assignments. The configuration space is actually the \emph{chessboard complex} $\mathsf{\Delta}_{2q-1,q}$, which is guaranteed to be $(q-2)$-connected. Roughly speaking, mapping this sufficiently connected space $X=\mathsf{\Delta}_{2q-1,q}$ to the $(q-1)$-dimensional simplex recording the popularity among the $q$ pieces, Volovikov's theorem shows that the center of the simplex cannot be missed, and thus that an almost equal popularity of the pieces can be achieved. With this technique, the existence of an envy-free multi-division is shown for a general class of preferences that are not necessarily monotone. 

\subsection{Tools from equivariant topology}\label{sec:tools}
We introduce now a specific abstract simplicial complex that will play a central role in the proof of Theorem~\ref{thm:group-birth}. The {\em chessboard complex} $\mathsf{\Delta}_{2q-1,q}$ is the abstract simplicial complex whose ground set is $[2q-1] \times [q]$ and whose simplices are the subsets $\sigma \subseteq [2q-1] \times [q]$ such that for every two distinct pairs $(r,j)$ and $(r',j')$ in $\sigma$ we have $r\neq r'$ and $j\neq j'$. The name comes from the following: If we interpret $[2q-1] \times [q]$ as a $(2q-1)\times q$ chessboard, the simplices are precisely the configurations of pairwise non-attacking rooks. See Figures \ref{subfig:chessboard1} and \ref{subfig:chessboard2} for an illustration of the chessboard complex $\mathsf{\Delta}_{2q-1,q}$ when $q=2$ and $q=3$.

Given an additive group $G$ of order $q$, we get a natural action $(\varphi_g)_{g\in G}$ of $G$ on $\mathsf{\Delta}_{2q-1,q}$ by identifying $[q]$ with $G$ via a bijection $\eta\colon [q]\rightarrow G$: this natural action is defined by $\varphi_g(r,j)=g\cdot(r,j)\coloneqq(r, \eta^{-1}(g+\eta(j)))$. 
This action is {\em free}, namely the orbit of each point in any geometric realization of $\mathsf{\Delta}_{2q-1,q}$ is of size $q$. Equivalently, the relative interiors of $\varphi_g(\sigma)$ and $\sigma$ are disjoint for every simplex $\sigma$ of $\mathsf{\Delta}_{2q-1,q}$ and every element $g$ of $G$ distinct from the neutral element; see~\cite[Chapter 6]{matouvsek2003using}.

The following lemma is an immediate consequence of Volovikov's theorem~\cite{volovikov1996topological}, which has found many applications in topological combinatorics. For an additive group $G$, we denote by $\Delta^G$ the standard simplex whose vertices are the unit vectors $\e_g$ (where $g$ ranges over $G$) of $\R^G$. The group acts naturally on $\Delta^G$ by setting $\varphi_{g'}(\e_g) \coloneqq \e_{g+g'}$ and by extending the action affinely on each face of $\Delta^G$.
When $G$ does not satisfy the condition of the lemma, its conclusion does not necessarily hold. Already for $G=\Z_6$, counterexamples are known; see~\cite{vzivaljevic1998user}.

(Given two topological spaces $X$ and $Y$ on which the group $G$ acts, a map $h\colon X\rightarrow Y$ is {\em $G$-equivariant} if $h(g\cdot x)=g\cdot h(x)$ for all $ x\in X$ and all $g\in G$.)

\begin{lemma}\label{lem:vol-surj}
Let $q=p^k$, where $p$ is a prime number and $k$ a positive integer. Denote by $G$ the additive group $\left((\Z_p)^k,+\right)$. For any $G$-equivariant continuous map $f\colon \mathsf{\Delta}_{2q-1,q} \rightarrow \Delta^G$, there exists $\x_0\in \|\mathsf{\Delta}_{2q-1,q}\|$ such that 
$f(\x_0)=\frac 1 q \sum_{g\in G}\e_g$.
\end{lemma}
\begin{proof}
Volovikov's theorem states the following; see~\cite[Section 6.2, Notes]{matouvsek2003using}. {\em Let $G$ be the additive group $\left((\Z_p)^k,+\right)$ and let $X$ and $Y$ be two topological spaces on which $G$ acts in a fixed-point free way. If $X$ is $d$-connected and $Y$ is a $d$-dimensional sphere, then there is no $G$-equivariant continuous map $X\rightarrow Y$.} An action is {\em fixed-point free} if each orbit has at least two elements. The simplicial complex $\mathsf{\Delta}_{2q-1,q}$ is $(q-2)$-connected (see~\cite{bjorner1994chessboard}) and the action is fixed-point free because it is free. The simplicial complex $\partial\Delta^G$ is a $(q-2)$-dimensional sphere and the action is fixed-point free as we explain now. It is enough to show that for every simplex $\sigma$ there is an element $g \in G$ such that $g \cdot \sigma$ is distinct from $\sigma$ because by considering the support of any point $\x$, this would show that the image of $\x$ by $g$ would be distinct from $\x$. Consider any simplex $\sigma$, pick an arbitrary vertex $g'\in G$ of $\sigma$ and an arbitrary vertex $g''\in G$ of $\partial\Delta^G$ not in $\sigma$ (which exists because the full simplex is not present). Set $g\coloneqq g''-g'$. The image of $\sigma$ by $g$ contains $g''$ and is thus a simplex distinct from $\sigma$.

Suppose for a contradiction that the image of ${f}$ misses $\frac 1 q \sum_{g\in G}\e_g$, which is actually the barycenter of $\Delta^G$. Define then $h(\x)$ as the intersection of $\partial \Delta^G$ with the half-line originating at the barycenter of $\Delta^G$ and going through $f(\x)$.  
This map contradicts Volovikov's theorem.
\end{proof}

\subsection{Encoding divisions via the chessboard complex}\label{subsec:encode-chess}
The proof of Theorem~\ref{thm:group-birth} uses a ``configuration space'' encoding some possible contiguous and feasible multi-divisions with $q-1$ long knives. 
This configuration space is the simplicial complex $\mathsf{\Delta}_{2q-1,q}$ introduced in Section~\ref{sec:tools}, with $q=p^k$ and $G=\left((\Z_p)^k,+\right)$ acting on it. The elements in $G$ will be used to identify the layered pieces. We choose an arbitrary bijection $\eta\colon [q] \rightarrow G$ to ease this identification. (In case $k=1$, it is certainly most intuitive to set $\eta(j)=j$; note that when $k\neq 1$, this definition does not make sense.) Moreover, we fix an arbitrary injective map
$h\colon [m] \rightarrow G$ and a geometric realization of $\mathsf{\Delta}_{2q-1,q}$. We denote by $\vv_{r,j}$ the realization of the vertex $(r,j)$. We explain now how each point of $\|\mathsf{\Delta}_{2q-1,q}\|$ encodes a multi-division with $q-1$ long knives.

We assign to each point $\x$ of $\|\mathsf{\Delta}_{2q-1,q}\|$ a $(q-1)$-dimensional simplex of $\mathsf{\Delta}_{2q-1,q}$ containing it. A tie can occur, e.g., for $q=3$, when $\x$ belongs to the interior of the edge $\vv_{3,1}$-$\vv_{1,2}$; there are three triangles containing this edge in $\mathsf{\Delta}_{5,3}$; each of them contains the vertices $\vv_{3,1}$ and $\vv_{1,2}$; the third vertex can be any of  $\vv_{2,3}$, $\vv_{4,3}$, and $\vv_{5,3}$. We make this assignment in such a way that all points with same support are assigned to the same $(q-1)$-dimensional simplex. Moreover, we make this assignment ``equivariant'': given any $g\in G$, the simplex assigned to $g\cdot\x$ is the image by $\varphi_g$ of the simplex assigned to $\x$. This is possible because the action of $G$ on $\mathsf{\Delta}_{2q-1,q}$ is free. We call this assignment the procedure $\mathfrak{P}$.

Consider any point $\x$ in $\|\mathsf{\Delta}_{2q-1,q}\|$. Let $\vv_{r_1,1}, \vv_{r_2,2}, \ldots, \vv_{r_q,q}$ be the vertices of the simplex assigned to $\x$ by $\mathfrak{P}$. We write then $\x$ as $\sum_{j=1}^q x_{r_j} \vv_{r_j,j}$. We set $x_k=0$ for every $k\notin\{r_1,\ldots,r_q\}$. The values of $x_1,\ldots,x_{2q-1}$ do not depend on the choice made by $\mathfrak{P}$: only the vertices $\vv_{r,j}$ spanning the minimal simplex of $\mathsf{\Delta}_{2q-1,q}$ containing $\x$ get nonzero coefficients, and these coefficients are then the barycentric coordinates in this face.

Let $\rho$ be the permutation in $\mathcal{S}_q$ such that $r_{\rho(1)} < r_{\rho(2)} < \cdots < r_{\rho(q)}$. We interpret $x_{r_{\rho(j)}}$ as the length of the $j$-th piece: in a way similar to the traditional encoding of the divisions (see, e.g.,~\cite{su1999rental}), the $j$-th piece in any layer is of length $x_{r_{\rho(j)}}$. We give then the $j$-th piece of the $\ell$-th layer the element $\eta(\rho(j))+h(\ell)$ of $G$ as its ``bundle-name.'' We get a non-overlapping layered piece by considering all pieces with a same bundle-name: 
if it were overlapping, it would contain two pieces named $\eta(\rho(j))+h(\ell)=\eta(\rho(j))+h(\ell')$ with $\ell \neq \ell'$, which is not possible because $h$ is injective.  
The non-overlapping layered pieces obtained this way form the multi-division encoded by $\x$, which we denote by $\calA(\x)=(\calA_1(\x),\calA_2(\x),\ldots,\calA_q(\x))$, where $\calA_{j}(\x)$ is the layered piece formed by the pieces with bundle-name $\eta(j)\in G$. Clearly, each $\calA(\x)$ is a feasible and contiguous multi-division that uses $q-1$ long knives. 
Figures \ref{fig:chessboard1}, \ref{fig:chessboard2}, and~\ref{fig:chessboard3} illustrate the chessboard complex $\mathsf{\Delta}_{2q-1,q}$ for $q=2$ and $q=3$ and associated multi-divisions corresponding to points of $\|\mathsf{\Delta}_{2q-1,q}\|$. 

%

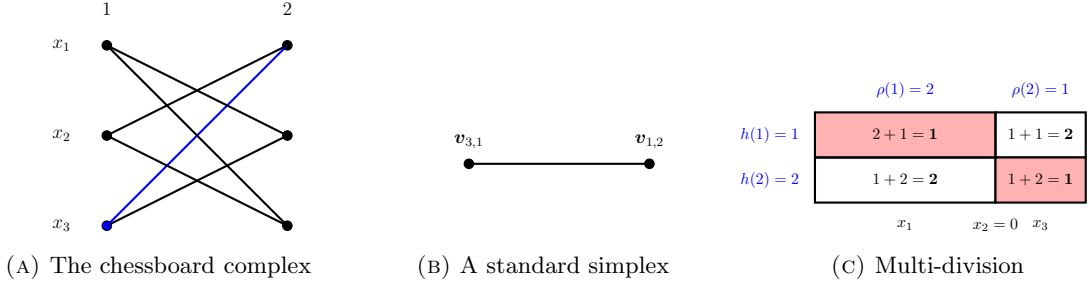
\begin{figure*}[htb]
\centering
\begin{subfigure}[t]{.33\linewidth}
\begin{tikzpicture}[scale=0.6, transform shape]
\node at (0,0) {};
\begin{scope}[xshift=3cm]
    \draw[fill=blue] (0,0) circle (3pt);
    \draw[fill=black] (0,2) circle (3pt);
    \draw[fill=black] (0,4) circle (3pt);
    
    \draw[fill=black] (4,0) circle (3pt);
    \draw[fill=black] (4,2) circle (3pt);
    \draw[fill=black] (4,4) circle (3pt);
    
    
    \draw[thick] (0,0) -- (4,2); 
    \draw[thick,blue] (0,0) -- (4,4);
    \draw[thick] (0,2) -- (4,0); 
    \draw[thick] (0,2) -- (4,4);
    \draw[thick] (0,4) -- (4,0); 
    \draw[thick] (0,4) -- (4,2);

    \node at (0,4.8) {\Large $1$};
    \node at (4,4.8) {\Large $2$};
    \node at (-1,4) {\Large $x_1$}; 
    \node at (-1,2) {\Large $x_2$}; 
    \node at (-1,0) {\Large $x_3$}; 
        
\end{scope}
\end{tikzpicture}
\subcaption{The chessboard complex}\label{subfig:chessboard1}
\end{subfigure}%
\begin{subfigure}[t]{.33\linewidth}
    \begin{tikzpicture}[scale=0.6, transform shape]  
    \node at (0,0) {};
\begin{scope}[xshift=-3.5cm,yshift=-2.5cm]  
    \draw[fill=black] (6,4) circle (3pt);
    \draw[fill=black] (10,4) circle (3pt);
    \draw[thick] (6,4) -- (10,4);
    \node at (6,4.5) {\Large $\vv_{3,1}$};
    \node at (10,4.5) {\Large $\vv_{1,2}$};
\end{scope}
\end{tikzpicture}
\subcaption{A standard simplex}\label{subfig:simplex1}
\end{subfigure}%
\begin{subfigure}[t]{.33\linewidth}
    \begin{tikzpicture}[scale=0.6, transform shape]
\begin{scope}[xshift=15cm,yshift=2cm]  
\draw[thick,fill=red!30] (0,0) rectangle (4,1);
        \draw[thick] (4,0) rectangle (6,1);
        \draw[thick] (0,-1) rectangle (4,0);
        \draw[thick,fill=red!30] (4,-1) rectangle (6,0);
        
        \node[blue] at (2,1.5) {$\rho(1)=2$};
        \node[blue] at (5,1.5) {$\rho(2)=1$};
        
        \node[blue] at (-1,0.5) {$h(1)=1$};
        \node[blue] at (-1,-0.5) {$h(2)=2$};
    
        \node at (2,0.5) {$2+1={\bf 1}$};
        \node at (5,0.5) {$1+1={\bf 2}$};
        \node at (2,-0.5) {$1+2={\bf 2}$};
        \node at (5,-0.5) {$1+2={\bf 1}$};
        
        \node at (2,-1.5) {$x_1$};
        \node at (4,-1.5) {$x_2=0$};
        \node at (5,-1.5) {$x_3$};
\end{scope}
\end{tikzpicture}
\subcaption{Multi-division}\label{subfig:multi-division1}
\end{subfigure}%
\caption{Illustration of the chessboard complex $\mathsf{\Delta}_{2q-1,q}$ when $q=2$. The blue edge of Figure~\ref{subfig:chessboard1} corresponds to a standard simplex, depicted on Figure~\ref{subfig:simplex1}, which represents the set of multi-divisions of the form described in Figure~\ref{subfig:multi-division1} when $m=2$ and $h(\ell)=\ell$ for the $\ell$-th layer. 
In Figure~\ref{subfig:multi-division1}, we have $G=(\Z_2,+)$ and $r_{\rho(1)}=r_2=1 < r_{\rho(2)}=r_1=3$.}
\label{fig:chessboard1}
\end{figure*}%

The multi-divisions $\calA(\x)$ enjoy some ``equivariant'' property.

\begin{lemma}\label{lem:eq-multi}
We have $\calA_j(\x) = \calA_{\eta^{-1}(g+\eta(j))}(g\cdot\x)$ for all $\x\in\|\mathsf{\Delta}_{2q-1,q}\|$, $j\in[q]$, and $g\in G$.
\end{lemma}
\begin{proof}
Let $\x\in\|\mathsf{\Delta}_{2q-1,q}\|$ and $g\in G$. Write $\x$ as $\sum_{j=1}^q x_{r_j} \vv_{r_j,j}$. By the definition of the action of $G$ on $\mathsf{\Delta}_{2q-1,q}$, we have $g\cdot\vv_{r,j} = \vv_{r,\eta^{-1}(g+\eta(j))}$ and thus we have
\begin{equation}\label{eq:gx}
g\cdot\x = \sum_{j=1}^q x_{r'_j}\vv_{r'_j,j},\quad\text{with}\quad r'_j= r_{\eta^{-1}(-g+\eta(j))}\, .
\end{equation}
(Note that the sets of $q$ distinct indices $r_j$ and $r'_{j'}$ are the same up to a permutation.) The equivariance of $\mathfrak{P}$ allows the vertices of the simplex chosen by $\mathfrak{P}$ for $g\cdot\x$ to be the $\vv_{r'_j,j}$. Denoting by $\rho'$ the permutation such that $r'_{\rho'(1)}<\cdots<r'_{\rho'(q)}$, the $j$-th piece of layer $\ell$ gets thus $\eta(\rho'(j))+h(\ell)$ as bundle-name in $\calA(g\cdot\x)$. Since the sets of indices $r_j$ and $r'_{j'}$ are the same, we have $r_{\rho(j)}=r'_{\rho'(j)}$ for all $j$, which implies with Equation~\eqref{eq:gx} that
$\rho(j)=\eta^{-1}\big(-g+\eta(\rho'(j))\big)$. 
Hence, the $j$-th piece of layer $\ell$ gets $g+\eta(\rho(j))+h(\ell)$ as bundle-name in $\calA(g\cdot\x)$.
Moreover, in both multi-divisions $\calA(\x)$ and $\calA(g\cdot\x)$, 
the $j$-th piece of layer $\ell$ is of the same length $x_{r_{\rho(j)}}=x_{r'_{\rho'(j)}}$. 
The layered piece with bundle-name $\eta(j)$ for $j\in[q]$ in $\calA(\x)$ gets thus bundle-name $g+\eta(j)$ in $\calA(g\cdot\x)$.
\end{proof}

See Figures~\ref{subfig:multi-division2} and~\ref{subfig:multi-division3} for an illustration of multi-divisions $\calA(\x)$ and $\calA(g\cdot\x)$. 

\begin{figure*}[htbt]
\centering
\begin{subfigure}[t]{.26\linewidth}
\begin{tikzpicture}[scale=0.6, transform shape]
    \node at (0,0) {};
\begin{scope}[xshift=2cm,yshift=0.5cm]
    \draw[fill=black] (0,0) circle (3pt);
    \draw[fill=black] (0,1) circle (3pt);
    \draw[fill=black] (0,2) circle (3pt);
    \draw[fill=black] (0,3) circle (3pt);
    \draw[fill=black] (0,4) circle (3pt);
    \draw[fill=black] (2,0) circle (3pt);
    \draw[fill=black] (2,1) circle (3pt);
    \draw[fill=black] (2,2) circle (3pt);
    \draw[fill=black] (2,3) circle (3pt);
    \draw[fill=black] (2,4) circle (3pt);     
    \draw[fill=black] (4,0) circle (3pt);
    \draw[fill=black] (4,1) circle (3pt);
    \draw[fill=black] (4,2) circle (3pt);
    \draw[fill=black] (4,3) circle (3pt);
    \draw[fill=black] (4,4) circle (3pt);     
    \node at (0,4.8) {\Large $1$};
    \node at (2,4.8) {\Large $2$};
    \node at (4,4.8) {\Large $3$};
    \node at (-1,4) {\Large $x_1$}; 
    \node at (-1,3) {\Large $x_2$}; 
    \node at (-1,2) {\Large $x_3$}; 
    \node at (-1,1) {\Large $x_4$}; 
    \node at (-1,0) {\Large $x_5$}; 
    \draw[thick,blue] (0,1) -- (4,2) -- (2,0) -- (0,1);
\end{scope}
\end{tikzpicture}
\subcaption{The chessboard complex}\label{subfig:chessboard2}
\end{subfigure}%
\begin{subfigure}[t]{.28\linewidth}
    \begin{tikzpicture}[scale=0.5, transform shape]  
\node at (0,0) {};
\begin{scope}[xshift=-6.5cm,yshift=0cm]
    \draw[fill=black] (8,0) circle (3pt);
    \draw[fill=black] (14,0) circle (3pt);
    \draw[fill=black] (11,5) circle (3pt);
    \draw[thick] (8,0) -- (14,0) -- (11,5) -- (8,0);

    \draw[fill=black] (11,2.5) circle (3pt);
    \node at (11,3) {\Large $\x$};
    \node at (7.5,0.5) {\Large $\vv_{5,2}$};
    \node at (14.5,0.5) {\Large $\vv_{4,1}$}; 
    \node at (11,5.5) {\Large $\vv_{3,3}$}; 
\end{scope}
\end{tikzpicture}
\subcaption{A standard simplex}\label{subfig:simplex2}
\end{subfigure}%
\begin{subfigure}[t]{.4\linewidth}
    \begin{tikzpicture}[scale=0.58, transform shape]      

    \begin{scope}[xshift=20cm,yshift=2cm]
            \draw[thick,fill=red!30] (0,0) rectangle (4,1);
            \draw[thick,fill=blue!30] (4,0) rectangle (6,1);
            \draw[thick] (6,0) rectangle (8,1);
            
            \draw[thick,fill=blue!30] (0,-1) rectangle (4,0);
            \draw[thick] (4,-1) rectangle (6,0);
            \draw[thick,fill=red!30] (6,-1) rectangle (8,0);
            
            \draw[thick] (0,-2) rectangle (4,-1);
            \draw[thick,fill=red!30] (4,-2) rectangle (6,-1);
            \draw[thick,fill=blue!30] (6,-2) rectangle (8,-1);
            
            \node at (2,-2.5) {$x_3$};
            \node at (0,-2.5) {$x_1=x_2=0$};
            \node at (5,-2.5) {$x_4$};
            \node at (7,-2.5) {$x_5$};
            
            \node[blue] at (2,1.5) {$\rho(1)=3$};
            \node[blue] at (5,1.5) {$\rho(2)=1$};
            \node[blue] at (7,1.5) {$\rho(3)=2$};
            
            \node[blue] at (-1,0.5) {$h(1)=1$};
            \node[blue] at (-1,-0.5) {$h(2)=2$};
            \node[blue] at (-1,-1.5) {$h(3)=3$};
            
            \node at (2,0.5) {$3+1={\bf 1}$};
            \node at (5,0.5) {$1+1={\bf 2}$};
            \node at (7,0.5) {$2+1={\bf 3}$};
            
            \node at (2,-0.5) {$3+2={\bf 2}$};
            \node at (5,-0.5) {$1+2={\bf 3}$};
            \node at (7,-0.5) {$2+2={\bf 1}$};

            \node at (2,-1.5) {$3+3={\bf 3}$};
            \node at (5,-1.5) {$1+3={\bf 1}$};
            \node at (7,-1.5) {$2+3={\bf 2}$};            
            
        \end{scope}    
\end{tikzpicture}
\subcaption{Multi-division~$\mathcal{A}(\x)$}\label{subfig:multi-division2}
\end{subfigure}%
\caption{Illustration of the division encoding and of the statement of Lemma~\ref{lem:eq-multi} in the case of $G=(\Z_3,+)$. 
Figure~\ref{subfig:chessboard2} is an illustration of the chessboard complex $\mathsf{\Delta}_{2q-1,q}$. The blue edges of Figure~\ref{subfig:chessboard2} correspond to a standard simplex, depicted on Figure~\ref{subfig:simplex2}, which represents multi-divisions of the form described in Figure~\ref{subfig:multi-division2}. In Figure~\ref{subfig:multi-division2}, $r_{\rho(1)}=r_3=3< r_{\rho(2)}=r_1=4 < r_{\rho(3)}=r_2=5$. The number in the $j$-th piece of the $\ell$-th layer corresponds to $\eta(\rho(j))+h(\ell)$, where $\eta$ is the identity map (because when $k=1$, we can identify $(\Z_p)^k$ and $[p]$). }
\label{fig:chessboard2}
\end{figure*}
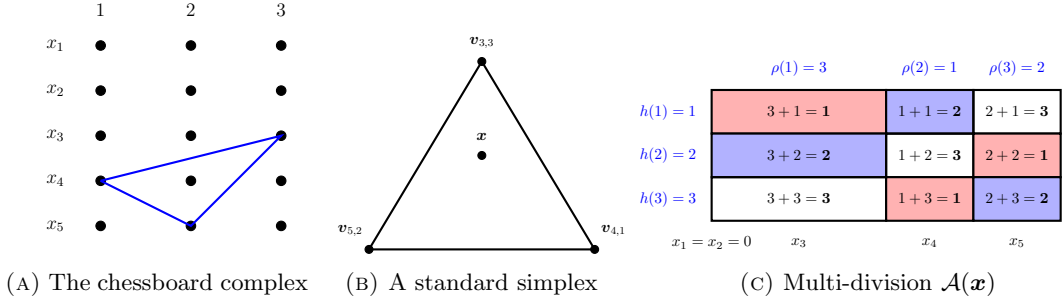

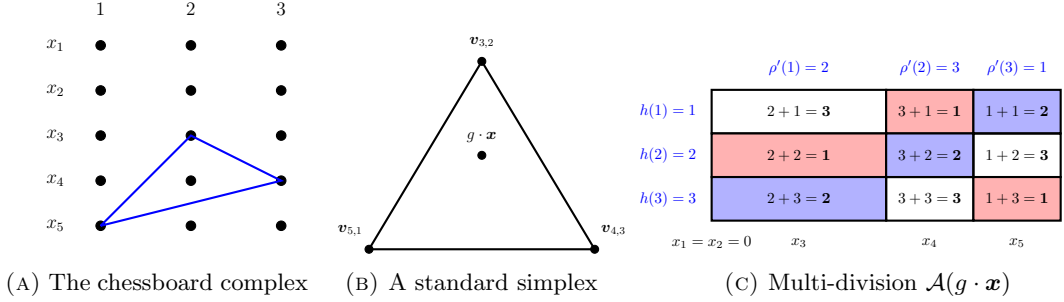
\begin{figure*}[htbt]
\centering
\begin{subfigure}[t]{.26\linewidth}
\begin{tikzpicture}[scale=0.6, transform shape]
    \node at (0,0) {};
\begin{scope}[xshift=2cm,yshift=0.5cm]
    \draw[fill=black] (0,0) circle (3pt);
    \draw[fill=black] (0,1) circle (3pt);
    \draw[fill=black] (0,2) circle (3pt);
    \draw[fill=black] (0,3) circle (3pt);
    \draw[fill=black] (0,4) circle (3pt);
    \draw[fill=black] (2,0) circle (3pt);
    \draw[fill=black] (2,1) circle (3pt);
    \draw[fill=black] (2,2) circle (3pt);
    \draw[fill=black] (2,3) circle (3pt);
    \draw[fill=black] (2,4) circle (3pt);     
    \draw[fill=black] (4,0) circle (3pt);
    \draw[fill=black] (4,1) circle (3pt);
    \draw[fill=black] (4,2) circle (3pt);
    \draw[fill=black] (4,3) circle (3pt);
    \draw[fill=black] (4,4) circle (3pt);     
    \node at (0,4.8) {\Large $1$};
    \node at (2,4.8) {\Large $2$};
    \node at (4,4.8) {\Large $3$};
    \node at (-1,4) {\Large $x_1$}; 
    \node at (-1,3) {\Large $x_2$}; 
    \node at (-1,2) {\Large $x_3$}; 
    \node at (-1,1) {\Large $x_4$}; 
    \node at (-1,0) {\Large $x_5$}; 
    \draw[thick,blue] (0,0) -- (2,2) -- (4,1) -- (0,0);
\end{scope}
\end{tikzpicture}
\subcaption{The chessboard complex}\label{subfig:chessboard3}
\end{subfigure}%
\begin{subfigure}[t]{.28\linewidth}
    \begin{tikzpicture}[scale=0.5, transform shape]  
\node at (0,0) {};
\begin{scope}[xshift=-6.5cm,yshift=0cm]
    \draw[fill=black] (8,0) circle (3pt);
    \draw[fill=black] (14,0) circle (3pt);
    \draw[fill=black] (11,5) circle (3pt);
    \draw[thick] (8,0) -- (14,0) -- (11,5) -- (8,0);

    \draw[fill=black] (11,2.5) circle (3pt);
    \node at (11,3) {\Large $g \cdot \x$};
    \node at (7.5,0.5) {\Large $\vv_{5,1}$};
    \node at (14.5,0.5) {\Large $\vv_{4,3}$}; 
    \node at (11,5.5) {\Large $\vv_{3,2}$}; 
\end{scope}
\end{tikzpicture}
\subcaption{A standard simplex}\label{subfig:simplex3}
\end{subfigure}%
\begin{subfigure}[t]{.4\linewidth}
    \begin{tikzpicture}[scale=0.58, transform shape]     

    \begin{scope}[xshift=20cm,yshift=2cm]
            \draw[thick] (0,0) rectangle (4,1);
            \draw[thick,fill=red!30] (4,0) rectangle (6,1);
            \draw[thick,fill=blue!30] (6,0) rectangle (8,1);
            
            \draw[thick,fill=red!30] (0,-1) rectangle (4,0);
            \draw[thick,fill=blue!30] (4,-1) rectangle (6,0);
            \draw[thick] (6,-1) rectangle (8,0);
            
            \draw[thick,fill=blue!30] (0,-2) rectangle (4,-1);
            \draw[thick] (4,-2) rectangle (6,-1);
            \draw[thick,fill=red!30] (6,-2) rectangle (8,-1);
            
            \node at (2,-2.5) {$x_3$};
            \node at (0,-2.5) {$x_1=x_2=0$};
            \node at (5,-2.5) {$x_4$};
            \node at (7,-2.5) {$x_5$};
            
            \node[blue] at (2,1.5) {$\rho'(1)=2$};
            \node[blue] at (5,1.5) {$\rho'(2)=3$};
            \node[blue] at (7,1.5) {$\rho'(3)=1$};
            
            \node[blue] at (-1,0.5) {$h(1)=1$};
            \node[blue] at (-1,-0.5) {$h(2)=2$};
            \node[blue] at (-1,-1.5) {$h(3)=3$};
            
            \node at (2,0.5) {$2+1={\bf 3}$};
            \node at (5,0.5) {$3+1={\bf 1}$};
            \node at (7,0.5) {$1+1={\bf 2}$};
            
            \node at (2,-0.5) {$2+2={\bf 1}$};
            \node at (5,-0.5) {$3+2={\bf 2}$};
            \node at (7,-0.5) {$1+2={\bf 3}$};

            \node at (2,-1.5) {$2+3={\bf 2}$};
            \node at (5,-1.5) {$3+3={\bf 3}$};
            \node at (7,-1.5) {$1+3={\bf 1}$};            
            
        \end{scope}    
\end{tikzpicture}
\subcaption{Multi-division~$\mathcal{A}(g \cdot \x)$}\label{subfig:multi-division3}
\end{subfigure}%
\caption{
Illustration of the division encoding and of the statement of Lemma~\ref{lem:eq-multi} in the case of $G=(\Z_3,+)$. 
The blue edges of Figure~\ref{subfig:chessboard3} correspond to the image of the standard simplex of Figure~\ref{subfig:chessboard2} when $g=2$. This image is depicted on Figure~\ref{subfig:simplex3} and corresponds to multi-divisions of the form described in Figure~\ref{subfig:multi-division3}. In Figure~\ref{subfig:multi-division3}, $r'_{\rho'(1)}=r'_2=3< r'_{\rho'(2)}=r'_3=4 < r'_{\rho'(3)}=r'_1=5$ (with the notation of the proof). The number in the $j$-th piece of the $\ell$-th layer corresponds to $\eta(\rho'(j))+h(\ell)$, where $\eta$ is the identity map (because when $k=1$, we can identify $(\Z_p)^k$ and $[p]$).
}
\label{fig:chessboard3}
\end{figure*}

Another important point is that $\calA(\x)$ depends continuously on $\x$ as stated by the following lemma. The convergence of multi-divisions is defined according to the pseudo-metric $d(\cdot,\cdot)$.

\begin{lemma}\label{lem:converg}
Let $\left(\x^{(t)}\right)_{t\in\Z_+}$ be a sequence of points of $\|\mathsf{\Delta}_{2q-1,q}\|$ converging to some limit point $\x^{(\infty)}$. Then $\left(\calA(\x^{(t)})\right)_{t\in\Z_+}$ converges to $\calA(\x^{(\infty)})$.
\end{lemma}
\begin{proof}
Consider a $(q-1)$-dimensional simplex $\sigma$ of $\mathsf{\Delta}_{2q-1,q}$ whose interior contains infinitely many $\x^{(t)}$. (There can actually be several simplices of this kind---this is what makes the proof not completely obvious---but they all contain $\x^{(\infty)}$.) Denote by $\vv_{r_1,1}, \vv_{r_2,2}, \ldots, \vv_{r_q,q}$ its vertices. Let $\rho$ be the permutation such that $r_{\rho(1)} < r_{\rho(2)} < \cdots < r_{\rho(q)}$.
We will define soon a multi-division $\widetilde \calA^{\sigma}$ such that $d{\left(\widetilde\calA^{\sigma},\calA(\x^{(\infty)})\right)}=0$, and that
$\lim d{\left(\calA(\x^{(t)}),\widetilde\calA^{\sigma}\right)}=0$, where the limit is taken over the $t$ for which $\x^{(t)}\in\sigma$. Since this will hold for every simplex $\sigma$ containing infinitely many $\x^{(t)}$ in its interior, we will get the desired conclusion.

The multi-division $\widetilde\calA^{\sigma}$ is defined as if $\sigma$ were the $(q-1)$-dimensional simplex chosen by $\mathfrak{P}$ for $\x^{(\infty)}$: the quantity $x_{r_{\rho(j)}}^{(\infty)}$ is the length of the $j$-th piece in any layer (note that $\rho$ depends on $\sigma$); the $j$-th piece of the $\ell$-th layer gets $\eta(\rho(j))+h(\ell)$ as a bundle-name;
$\widetilde \calA_{j}^{\sigma}$ is obtained by taking all pieces with bundle-name $\eta(j)$. 
Denote by $L^j_{\ell}(\x^{(t)})$ the $j$-th piece in the $\ell$-th layer in multi-division $\calA(\x^{(t)})$ and by $\widetilde L_{\ell}^j$ the $j$-th piece in the $\ell$-th layer in multi-division $\widetilde\calA^{\sigma}$. We have 
\[
d{\left(L^j_{\ell}(\x^{(t)}), \widetilde L_{\ell}^j\right)}\le \left|\sum_{j'=1}^{j-1} \left(x_{r_{\rho(j')}}^{(t)}- x_{r_{\rho(j')}}^{(\infty)}\right)\right|+\left|\sum_{j'=1}^j \left(x_{r_{\rho(j')}}^{(t)}-x_{r_{\rho(j')}}^{(\infty)}\right)\right|
\]
(by comparing the positions of the left- and right-hand endpoints of $L^j_{\ell}(\x^{(t)})$ and $\widetilde L_{\ell}^j$), whose right-hand side converges to $0$ when $t$ goes to $+\infty$ by definition of the $x_i^{(\infty)}$'s. Since $L^j_{\ell}(\x^{(t)})$ and $\widetilde L_{\ell}^j$ get the same bundle-name, we have $\lim d{\left(\calA(\x^{(t)}),\widetilde\calA^{\sigma}\right)}=0$, where the limit is taken over the $t$ for which $\x^{(t)}\in\sigma$.

The point $\x^{(\infty)}$ belongs to the interior of a unique face $\tau$ of $\sigma$.
Let $\vv_{r'_1,1}, \vv_{r'_2,2}, \ldots, \vv_{r'_q,q}$ be the vertices of the $(n-1)$-dimensional simplex chosen by $\mathfrak{P}$ for $\x^{(\infty)}$. Denote by $J$ the set of indices $j$ for which $\vv_{r'_j,j}$ is a vertex of $\tau$. (Note that $|J|=\dim\tau +1$.) Since $\tau$ is unique and its vertices uniquely determined, we have $r_j = r'_j$ for all $j \in J$. Let $\rho'$ be the permutation so that $r'_{\rho'(1)} < r'_{\rho'(2)} < \cdots < r'_{\rho'(q)}$. Among them, we have the elements $r'_j$ with $j \in J$, which are the same as the elements $r_j$ with $j \in J$. Denoting by $j_1 < \cdots < j_{|J|}$ the elements from $\rho^{-1}(J)$ and by $j'_1 < \cdots < j'_{|J|}$ the elements from ${\rho'}{}^{-1}(J)$, we have thus $r_{\rho(j_a)} = r'_{\rho'(j'_a)}$ for every $a \in [|J|]$. This implies, since $r_j = r'_j$ for all $j\in J$, that $\rho(j_a) = \rho'(j'_a)$ for all $a\in[|J|]$. Therefore, the $a$-th piece of nonzero length (of any layer) in $\widetilde\calA^{\sigma}$ has length $x^{(\infty)}_{r_{\rho(j_a)}}$, which equals $x^{(\infty)}_{r'_{\rho'(j'_a)}}$, length of the $a$-th piece of nonzero length (of any layer) in $\widetilde\calA(\x^{(\infty)})$. Its bundle-name in $\widetilde\calA^{\sigma}$ is $\eta(\rho(j_a))+h(\ell)$, which is equal to $\eta(\rho'(j'_a))+h(\ell)$, its bundle-name in $\calA(\x^{(\infty)})$. Hence, the pieces of nonzero length in any layer are the same in $\calA(\x^{(\infty)})$ and $\widetilde\calA^{\sigma}$ and get the same bundle-names. The pseudo-metric $d(\cdot,\cdot)$ between two pieces of zero-length is $0$. Therefore, $d{\left(\widetilde\calA^{\sigma},\calA(\x^{(\infty)})\right)}=0$. 
\end{proof}

\subsection{Proof of Theorem~\ref{thm:group-birth}}


Before proceeding to the proof of Theorem~\ref{thm:group-birth}, we prove the following auxiliary lemma, which will also be used in Section~\ref{sec:twolayers}. In the applications of this lemma, the vertices of the graph $H$ will represent on one side the (non-birthday) agents and on the other side the pieces, and its edges will represent the acceptable assignments (in terms of preferences). We denote by $\delta_H(i)$ the edges incident to a vertex $i$ in the graph $H$.

\begin{lemma}\label{lem:TUM}
Let $n,q$ be positive integers. 
Let $a_1,a_2,\ldots,a_{q}$ be nonnegative real numbers summing up to $n-1$. 
Consider a bipartite graph $H=([n-1],[q];E)$ with nonnegative weights $w_e$ on its edges $e \in E$. Suppose $\sum_{e \in \delta_H(i)}w_e=1$ for each $i \in [n-1]$ and $\sum_{e \in \delta_H(j)}w_e=a_j$ for each $j \in [q]$. Then for every $j^* \in [q]$ there is an assignment $\pi_{j^*} \colon [n] \rightarrow [q]$ such that 
\begin{itemize}
    \item $\pi_{j^*}(n)=j^*$,
    \item for each $i \in [n-1]$, the vertex $\pi_{j^*}(i)$ is a neighbor of $i$ in $H$, 
    \item $|\pi_{j^*}^{-1}(j^*)| = \lfloor a_{j^*} \rfloor +1$, and
    \item for each $j \in [q]$, we have $|\pi_{j^*}^{-1}(j)| \in \{ \lfloor a_j \rfloor, \lceil a_j \rceil  \}$.
\end{itemize}
\end{lemma}
\begin{proof}
Take any $j^* \in [q]$. Consider the following polytope: 
\[
P= \left\{\, \x \in \mathbb{R}^E_{+} \middle| \sum_{e \in \delta_H(i)}x_{e}=1 \text{ for all } i \in [n-1]\ \mbox{and} \  \lfloor a_j \rfloor \le \sum_{e \in \delta_H(j)}x_{e} \le \lceil a_j \rceil \text{ for all } j \in [q]\, \right\} \, .
\]
The polytope $P$ is nonempty because it contains $\w=(w_e)_{e \in E}$.
Note that since the incidence matrix defining $P$ is totally unimodular, the vertices of $P$ are integral. Now, we claim that there exists an integral vertex ${\bar \w} \in P$ with $\sum_{e \in \delta_H(j^*)}{\bar w}_e= \lfloor a_{j^*} \rfloor$. To see this, if $a_{j^*}$ is an integer, the existence of such an integral vertex is obvious. If $a_{j^*}$ is not an integer, then just note that all integral vertices of $P$ cannot make the previous sum equal to $\lceil a_{j^*} \rceil$ since there exists a point in $P$ not satisfying this equality, namely $\w$.
The coefficient $\bar w_e$ belongs to $\{0,1\}$ for every $e\in E$ because $\sum_{e \in \delta_H(i)}{\bar w}_e=1$ for each $i \in [n-1]$ and because ${\bar \w}$ is integral. Thus every $i \in [n-1]$ has a unique neighbor $j \in [q]$ with ${\bar w}_{ij}=1$. 
Defining $\pi_{j^*}(i)$ as this $j$ for $i \in [n-1]$ and $\pi_{j^*}(n)$ as $j^*$, we get a desired assignment.~ 
\end{proof}

\smallskip 
Now we are ready to prove Theorem~\ref{thm:group-birth}. 
\smallskip 

\begin{proof}[Proof of Theorem~\ref{thm:group-birth}]
Without loss of generality, agent $n$ is the birthday agent. Let $p$ be the prime number and $k$ the integer such that $q=p^k$. Let $\T$ be a triangulation of $\mathsf{\Delta}_{2q-1,q}$ such that $\varphi_g(\sigma) \in \T$ for all $\sigma \in \T$ and all $g \in G$. (It is invariant by the action of $G=((\Z_p)^k,+)$.) Such a triangulation can be achieved by taking repeated barycentric subdivisions of $\mathsf{\Delta}_{2q-1,q}$.

We partition the vertices of $\T$ into their $G$-orbits. From each orbit, we pick a vertex $\vv$. We ask each non-birthday agent $i \in[n-1]$ the index $j$ of the non-overlapping layered piece $\calA_j(\vv)$ she prefers in $\calA(\vv)$. (In case of a tie, she makes an arbitrary choice.) We define $f^{(i)}(\vv)$ to be $\e_{\eta(j)}$. We extend $f^{(i)}$ on each orbit in an equivariant way: $f^{(i)}(g \cdot \vv)\coloneqq g \cdot f^{(i)}(\vv)$. This is done unambiguously because the action of $G$ on $\T$ is free. Lemma~\ref{lem:eq-multi} implies that, for every vertex $\vv$ of the triangulation $\T$, the integer $j$ such that $f^{(i)}(\vv)=\e_{\eta(j)}$ is the index of a layered piece preferred by agent $i$ in the multi-division $\calA(\vv)$.

For each non-birthday agent $i \in[n-1]$, we extend the map $f^{(i)}$ affinely on each simplex of $\T$. Denote by ${\bar f}^{(i)}$ the affine extension of $f^{(i)}$. This way, the map ${\bar f}^{(i)}$ is a $G$-equivariant simplicial map from $\T$ to $\Delta^G$. The affine extension ${\bar f}^{(i)}$  of $f^{(i)}$ is an ``approximation'' of the original preferences. 
We then aggregate these approximate preferences among $n-1$ non-birthday agents by setting ${\bar f}=\frac 1 {n-1} \sum_{i=1}^{n-1}{\bar f}^{(i)}$. It is a $G$-equivariant continuous map from $\mathsf{\Delta}_{2q-1,q}$ to $\Delta^G$. 
For each vertex $\vv$ of $\T$, the point ${\bar f}(\vv)$ represents the average preference of the $n-1$ agents according to the original preferences. 
For each $\x$ in $\|\mathsf{\Delta}_{2q-1,q}\|$, ${\bar f}(\x)$ represents the average preference of the $n-1$ agents according to the approximate preferences.
(This averaging technique has been introduced by Gale~\cite{gale1984equilibrium} and applied by Asada et al.~\cite{Asada2018} for the birthday cake-division.)

According to Lemma~\ref{lem:vol-surj}, there exists a point $\x_0$ in $\|\mathsf{\Delta}_{2q-1,q}\|$ such that 
\begin{equation}\label{eq:center}
{\bar f}(\x_0)=\frac 1 q \sum_{g\in G}\e_g\, .
\end{equation}

For every non-birthday agent $i \in [n-1]$ and every layered piece $j \in [q]$, define $w_{ij} = {\bar f}^{(i)}(\x_0) \cdot \e_{\eta(j)}$ 
(where the product in the right-hand term is the dot product in $\R^G$). 
We clearly have $\sum_{i=1}^{n-1} w_{ij}= \frac {n-1} q$ for all $j \in [q]$ (by equation~\eqref{eq:center}), and $\sum_{j=1}^q w_{ij}= 1$ for all $i \in [n-1]$ (because ${\bar f}^{(i)}$ has its image in $\Delta^G$). 

Now, consider the bipartite graph $H=([n-1],[q];E)$, with one side being the agents from $1$ to $n-1$ and with the other side being the layered pieces and where the edge $ij$ exists precisely when $w_{ij} > 0$. 
Applying Lemma~\ref{lem:TUM} with $a_j= \frac{n-1}{q}$ for $j \in [q]$, there exists for every $j^* \in [q]$ an assignment $\pi_{j^*} \colon [n] \rightarrow [q]$ such that 
\begin{itemize}
    \item $\pi_{j^*}(n)=j^*$, 
    \item for each $i \in [n-1]$, the vertex $\pi_{j^*}(i)$ is a neighbor of $i$ in $H$, 
    \item $|\pi^{-1}_{j^*}(n)| = \lfloor (n-1)/q \rfloor +1$, and
    \item for each $j \in [q]$, we have $|\pi^{-1}_{j^*}(j)| \in \{ \lfloor (n-1)/q \rfloor, \lceil (n-1)/q \rceil \}$.
\end{itemize}
We claim that 
$$
\lfloor n/q\rfloor \le |\pi^{-1}_{j^*}(j)| \le  \lceil n/q\rceil,
$$ 
for all $j\in [q]$.
Clearly, for each $j \in [q]$, we have $|\pi^{-1}_{j^*}(j)| \le \lceil n/q\rceil$ since $\lceil (n-1)/q\rceil \le \lceil n/q\rceil$ and $\lfloor (n-1)/q \rfloor +1=\lceil n/q \rceil$. To see the lower bound, 
observe that if $n$ is not a multiple of $q$, we have $\lfloor (n-1)/q \rfloor = \lfloor n/q \rfloor$ and hence $|\pi^{-1}_{j^*}(j)| \ge \lfloor n/q \rfloor$ for each $j \in [q]$; if $n$ is a multiple of $q$, we have $|\pi^{-1}_{j^*}(j)| \le n/q$ for each $j \in [q]$, which together with the fact that $\sum_{j \in [q]}|\pi^{-1}_{j^*}(j)|=n$ implies that $|\pi^{-1}_{j^*}(j)| = n/q$ for each $j \in [q]$. 

For every integer $N>0$, we can choose $\T \coloneqq \T_N$ so that it has a mesh size upper bounded by $1/N$ and define ${\bar f}^{(i)} \coloneqq {\bar f}^{(i)}_N$. For each $N$, we have $\x^{N}_{0}$ satisfying \eqref{eq:center}. 
Let $H_N$ be the graph $([n-1],[q];E_N)$ such that the edge $ij$ exists precisely when ${\bar f}^{(i)}_N(\x^N_0) \cdot \e_{\eta(j)} > 0$. Compactness implies that we can select among these arbitrarily large $N$ an infinite sequence such that $(\x^{N}_0)$ converges to a point $\x^*$ and such that $H_{N}$ is always the same graph $H^*$. As we have seen, for every $j^*\in[q]$, there is then an assignment $\pi_{j^*}\colon[n]\rightarrow[q]$ with:
\begin{itemize}
    \item $\pi_{j^*}(n)=j^*$,
    \item for each $i\in [n-1]$, the vertex $\pi_{j^*}(i)$ is a neighbor of $i$ in $H^*$, and
    \item for each $j\in [q]$, we have $|\pi^{-1}_{j^*}(j)|\in\{\lfloor n/q\rfloor, \lceil n/q\rceil \}$.
\end{itemize}

Consider any $j^* \in [q]$ and any $N$ from the infinite sequence. We have $H_N = H^*$. By definition of $\bar f^{(i)}_N$, for each agent $i \in [n-1]$, there exists thus a vertex $\vv^{i,j^*,N}$ of the supporting simplex of $\x^{N}_0$ in $\T_N$ such that $f^{(i)}_N(\vv^{i,j^*,N})$ is $\e_{\eta(\pi_{j^*}(i))}$, meaning that $\mathcal{A}_{\pi_{j^*}(i)}(\vv^{i,j^*,N}) \in c_i(\mathcal{A}(\vv^{i,j^*,N}))$. 
We have $\calA_{\pi_{j^*}(i)}(\vv^{i,j^*,N})\in c_i(\calA(\vv^{i,j^*,N}))$ for all $i\in[n-1]$ and arbitrarily large $N$. 
Since $(\x^{N}_0)$ converges to $\x^*$, the sequence $(\vv^{i,j^*,N})$ converges to the same point ${\x}^*$ for every $i \in [n-1]$ and every $j^* \in [q]$. 
By the closed preferences assumption and Lemma~\ref{lem:converg}, we have $\calA^*_{\pi_{j^*}(i)}\in c_i(\calA^*)$ for all $i\in[n-1]$, where $\calA^* \coloneqq \calA(\x^*)$. Since we have $\pi_{j^*}(n)=j^*$ and $|\pi^{-1}_{j^*}(j)|\in \{ \lfloor n/q \rfloor,\lceil n/q\rceil\}$ for all $j\in [q]$, the properties required for $\pi_{j^*}$ are satisfied. 
\end{proof}

\section{Envy-free division using one short knife and one long knife}\label{sec:twolayers}

While the main result of Section~\ref{sec:long} is a pure existence result with a non-constructive proof based on Volovikov's theorem, we focus in this section on the computational aspect of the problem of finding an envy-free multi-division. Actually, we show not only that an envy-free multi-division with one short and one long knife exists for a two-layered cake division among three groups of $n$ agents with closed, monotone, and hungry preferences, but also that such a division can be efficiently computed for agents with valuations satisfying the Lipschitz condition. 
In fact, we prove stronger statements where both existence and computational results extend to those for birthday multi-divisions.

We start by stating the existence result.

\begin{theorem}\label{thm:two-layers:group-birth}
Consider an instance of the two-layered cake-cutting problem with $n$ agents, $n \ge 3$, with closed, monotone, and hungry preferences. Then there exists a feasible and contiguous multi-division into three layered pieces so that no matter which layered piece the birthday agent chooses, the remaining agents can be assigned to the layered pieces while satisfying the following two properties:
\begin{itemize}
    \item each of the remaining agents is assigned to one of her preferred layered pieces.
    \item the number of agents assigned to each layered piece differs by at most one.
\end{itemize}
Moreover, this multi-division requires only one short knife and one long knife. 
\end{theorem}

In a more formal way (and with the birthday agent being labeled with $n$), Theorem~\ref{thm:two-layers:group-birth} ensures the existence of a feasible and contiguous multi-division $\calA$ into three layered pieces with the following property: for every $j^*\in\{1,2,3\}$, there is an assignment $\pi_{j^*}\colon[n]\rightarrow\{1,2,3\}$ with 
\begin{itemize}
    \item $\pi_{j^*}(n)=j^*$,
    \item for each $i\in [n-1]$, $\calA_{\pi_{j^*}(i)}\in c_i(\calA)$, and
    \item for each $j\in \{1,2,3\}$, $|\pi^{-1}_{j^*}(j)|\in\{\lfloor n/3\rfloor, \lceil n/3\rceil \}$.
\end{itemize}


For three agents with valuations satisfying monotonicity and the Lipschitz condition, Deng et al.~\cite{Deng2012} designed an FPTAS to compute an approximate envy-free division. 
We show that in the context of two-layered cake-cutting among $n$ agents, we can also design an FPTAS. 
This result, which can be seen as an algorithmic version of Theorem~\ref{thm:two-layers:group-birth}, generalizes the result of Deng et al. in three respects: First, our result holds for the birthday version; second, it holds for the group version, ensuring an envy-free division among three groups of almost equal size; third, it holds for the case of two-layered cake-cutting. Recall that monotonicity and the Lipschitz condition imply the hungry assumption under the valuation function model; see Section~\ref{sec:prelim}.


\begin{theorem}\label{thm:two-layers:group-birth:FPTAS}
Consider an instance of the two-layered cake-cutting problem with $n$ agents, $n\ge 3$, whose valuation functions satisfy monotonicity and the Lipschitz condition with constant $K$. Then, for any $\varepsilon>0$, one can find in time $O(n\log^2 \frac{K}{\varepsilon})$ a feasible and contiguous multi-division into three layered pieces where no matter which layered piece the birthday agent chooses, the remaining agents can be assigned to the layered pieces while satisfying the following two properties:
\begin{itemize}
    \item each of the remaining agents is assigned to one of her $\varepsilon$-approximate preferred layered pieces.
    \item the number of agents assigned to each layered piece differs by at most one.
\end{itemize}
Moreover, this multi-division requires only one short knife and one long knife. 
\end{theorem}

More formally (and with the birthday agent being labeled with $n$), Theorem~\ref{thm:two-layers:group-birth:FPTAS} ensures that one can find in time $O(n\log^2 \frac{K}{\varepsilon})$, a feasible and contiguous multi-division $\calA$ into three layered pieces with the following property: for every $j^*\in\{1,2,3\}$, there is an assignment $\pi_{j^*}\colon[n]\rightarrow\{1,2,3\}$ with
\begin{itemize}
    \item $\pi_{j^*}(n)=j^*$,
    \item for each $i\in [n-1]$, $v_i(\calA_{\pi_{j^*}(i)})+\varepsilon \ge \max_{j \in \{1,2,3\}}v_i(\calA_{j})$, and
    \item for each $j\in \{1,2,3\}$, $|\pi^{-1}_{j^*}(j)|\in\{\lfloor n/3\rfloor, \lceil n/3\rceil \}$.
\end{itemize}

Note that the algorithm of Theorem~\ref{thm:two-layers:group-birth:FPTAS} only computes the division and not the assignment itself. Yet, once we find a desired multi-division, we can ask each agent her approximate preferred pieces and then easily compute a desired assignment
$\pi_{j^*}$ in $O(n^2)$; this can be done, e.g., by a recent max-flow algorithm of Orlin~\cite{Orlin2013}. 

In order to establish the above theorems, we encode by the points of the unit square $[0,1]^2$ the divisions of the two-layered cake using the  short knife and the long knife. The position of the long knife corresponds to the $x$-axis and the position of the short knife corresponds to the $y$-axis. Fixing bundle names appropriately, the two vertical boundaries when $x=0$ and $x=1$ enjoy a certain symmetry: the divisions that appear on these boundaries are the same. By exploiting this symmetry, one can apply a Sperner-type argument to show the existence of an envy-free division. A careful utilization of the monotonicity allows to make this argument algorithmic and to get the FPTAS.

With this kind of techniques, we are able to prove the following algorithmic version of the theorem of Segal-Halevi and Suksompong~\cite{Segal2021}.

\begin{theorem}\label{thm:one-layer:group-birth:FPTAS}
Consider an instance of the one-layered cake-cutting problem with $n$ agents, $n\ge 3$, whose valuation functions satisfy monotonicity and the Lipschitz condition with constant $K$. Let $k_1, k_2, k_3$ be positive integers summing up to $n$. Then, for any $\varepsilon>0$, one can find in time $O(n\log^2 \frac{K}{\varepsilon})$ a feasible and contiguous multi-division into three pieces where no matter which piece the birthday agent chooses, the remaining agents can be assigned to the pieces while satisfying the following two properties:
\begin{itemize}
    \item each of the remaining agents is assigned to one of her $\varepsilon$-approximate preferred pieces.
    \item the number of agents assigned to piece $j$ is $k_j$ for all $j\in \{1,2,3\}$.
\end{itemize} 
\end{theorem}

More formally (and with the birthday agent being labeled with $n$), Theorem~\ref{thm:one-layer:group-birth:FPTAS} ensures that one can find in time $O(n\log^2 \frac{K}{\varepsilon})$, a feasible and contiguous multi-division $\calA$ into three pieces with the following property: for every $j^*\in\{1,2,3\}$, there is an assignment $\pi_{j^*}\colon[n]\rightarrow\{1,2,3\}$ with
\begin{itemize}
    \item $\pi_{j^*}(n)=j^*$,
    \item for each $i\in [n-1]$, $v_i(\calA_{\pi_{j^*}(i)})+\varepsilon \ge \max_{j \in \{1,2,3\}}v_i(\calA_{j})$, and
    \item for each $j\in \{1,2,3\}$, $|\pi^{-1}_{j^*}(j)|=k_j$.
\end{itemize}

We remark that the existence of such a division follows from the proof technique by Segal-Halevi and Suksompong~\cite{Segal2021} and the classical existence result of a birthday division~Woodall~\cite{woodall1980dividing} and Asada et al.~\cite{Asada2018}. However, the algorithmic result does not directly follow because the proof in~Segal-Halevi and Suksompong reduces the problem to the $n$-agent case (i.e., it invokes the $n-1$ dimensional version of Sperner's lemma).  



\subsection{Preliminaries}\label{subsec:preliminaries} In this section, we introduce the main tools and lemmas used for the proof of Theorems~\ref{thm:two-layers:group-birth},~\ref{thm:two-layers:group-birth:FPTAS}, and~\ref{thm:one-layer:group-birth:FPTAS}.

Consider $n$ agents who have closed, monotone, and hungry preferences. The $n$-th agent is assumed to be the birthday agent. We consider a two-layered cake. Its \emph{top layer} is the layer that is weakly preferred to the other layer, within the ``degenerate'' multi-division with no cuts, by the majority of agents; that is, at least $\lceil{\frac{n}{2}}\rceil$ agents weakly prefer the top layer to the other layer. This latter layer is then referred to as the \emph{bottom layer}.

As in Section~\ref{subsec:encode-chess}, we use a ``configuration space'' to encode the considered multi-divisions, but in a much simpler way. Here, the configuration space is the unit square. A point $(x,y)$ of the unit square encodes a division as in Figure~\ref{fig:EF:three} where $x$ corresponds to the long knife over the two layers and $y$ corresponds to the short knife over the top layer. We denote by $\mathcal{A}(x,y)=(\mathcal{A}_1(x,y),\mathcal{A}_2(x,y),\mathcal{A}_3(x,y))$ the multi-division in Figure~\ref{fig:EF:three} represented by $(x,y) \in [0,1]^2$ where the agents first cut the cake via the short knife position $y$, and then cut the rest via the long knife position $x$. Namely, $\mathcal{A}_1(x,y)$ consists the $[0,y]$ segment of the top layer, $\mathcal{A}_2(x,y)$ consists of the $[\max\{x,y\},1]$ of the top layer and the $[0,x]$ segment of the bottom layer, and $\mathcal{A}_3(x,y)$ contains the remaining pieces. The multi-division $\mathcal{A}(x,y)$ is feasible and contiguous.

\begin{figure}[htb]
\centering
    \begin{tikzpicture}[scale=0.6, transform shape]
    
    \begin{scope}[yshift=0cm]
        \draw[thick,fill=red!30] (0,0) rectangle (4,1);
        \draw[thick,fill=blue!30] (0,-1) rectangle (3,0); 
        \draw[thick,fill=blue!30] (4,0) rectangle (10,1); 
        \draw[thick] (3,-1) rectangle (10,0);

        \node at (2.0,0.5) {$1$};
        \node at (1.5,-0.5) {$2$};
        \node at (6.0,0.5) {$2$};
        \node at (6.0,-0.5) {$3$};

        \node at (3.0,1.3) {$x$};
        \node at (4.0,1.3) {$y$};
    \end{scope}
    
    \begin{scope}[yshift=3cm]
        \draw[thick,fill=red!30] (0,0) rectangle (1,1);
        \draw[thick,fill=blue!30] (0,-1) rectangle (3,0); 
        \draw[thick] (1,0) rectangle (3,1);
        \draw[thick,fill=blue!30] (3,0) rectangle (10,1); 
        \draw[thick] (3,-1) rectangle (10,0);

        \node at (0.5,0.5) {$1$};
        \node at (1.5,-0.5) {$2$};
        \node at (6.0,0.5) {$2$};
        \node at (2.0,0.5) {$3$};
        \node at (6.0,-0.5) {$3$};

        \node at (3.0,1.3) {$x$};
        \node at (1.0,1.3) {$y$};
    \end{scope}

    \end{tikzpicture}
\caption{Three-agent multi-divisions $\mathcal{A}(x,y)$. Note that the short knife $y$ is prioritized over the long knife $x$.}
\label{fig:EF:three}
\end{figure}
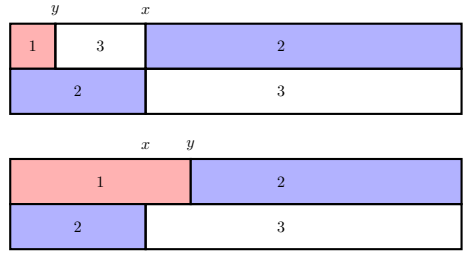

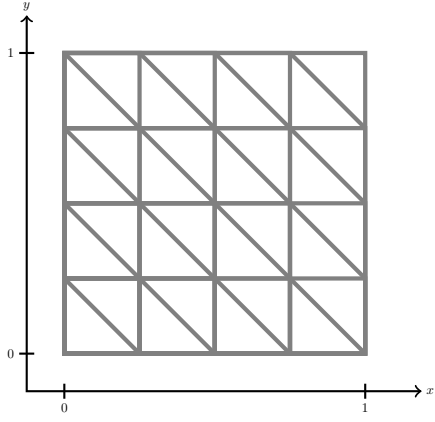
\begin{figure}
\centering
    \begin{tikzpicture}[scale=0.5, transform shape]

    \draw [<->,thick] (-1,9) node (yaxis) [above] {$y$}|- (9.5,-1) node (xaxis) [right] {$x$};
    \draw[thick] (-1.2,0) node [left] {$0$} -- (-0.8,0);
    \draw[thick] (-1.2,8) node [left] {$1$} -- (-0.8,8);

    \draw[thick] (0,-1.2) node [below] {$0$} -- (0,-0.8);
    \draw[thick] (8,-1.2) node [below] {$1$} -- (8,-0.8);
    
    \foreach \x in {0,2,4,6, 8}{
        \foreach \y in {0,2,4,6, 8}{
            \draw[color=gray,ultra thick] (0,0) rectangle (\x,\y);
        }
    }
    
    \foreach \x in {0, 2,4,6}{
           \foreach \y in {8,6,4,2}{
                \draw[color=gray,ultra thick] (\x, \y) -- (\x + 2, \y - 2);
            }
    }
    
\end{tikzpicture}
\caption{Triangulation of the unit square.}
\label{fig:unitsquare}
\end{figure}%

We divide the unit square into small squares. 
Let $N$ be an integer. 
Partition the unit square $[0,1]^2$ into $N^2$ smaller squares of side length $\frac 1N$; we call each of these squares a \emph{basic square}. 
We call each vertical line $\bigL(x)=\{(x,y) \mid y \in [0,1] \}$ of Figure \ref{fig:unitsquare} with $x=\frac{k}{N}$ for $k \in \{0,1,\ldots,N\}$ a \emph{basic line}. 

We construct a continuous map ${\bar f}$ from the unit square to the two-dimensional standard simplex, which represents the average preferences of the non-birthday agents. To this end, consider a particular triangulation $\T$ of the unit square; for each basic square, one new edge is added from the top left corner vertex to the bottom right corner vertex. See Figure \ref{fig:unitsquare}. We denote the vertices of the triangulation by $V(\T)$. 
We define $f^{(i)} \colon V(\T) \rightarrow \Delta^2$ which encodes the preferences of agent $i$. Namely, we ask each agent $i$ the index $j$ of the layered piece $\calA_j(\vv)$ she prefers in $\calA(\vv)$ and set $f^{(i)}(\vv)$ to be $\e_{j}$. In that case, we says that agent $i$ \emph{colors} the vertex $\vv$ with color $j$. We consider the following tie-breaking rule: a zero-length piece is never chosen and for each $\vv=(x,y)$, in case of a tie with piece $1$, each agent colors $\vv$ with $1$; in case of a tie between pieces $2$ and $3$ only, each agent colors $\vv$ with $2$ if $x \le y$ and with $3$ otherwise.

Similarly to the proof of Theorem~\ref{thm:group-birth}, we apply Gale's technique~\cite{gale1984equilibrium} to our problem. For each non-birthday agent $i \in[n-1]$, we extend the map $f^{(i)}$ affinely on each simplex of $\T$. Denote by ${\bar f}^{(i)}$ the affine extension of $f^{(i)}$.
We then aggregate these approximate preferences among $n-1$ non-birthday agents as follows:  define ${\bar f} \colon [0,1]^2 \rightarrow \Delta^2$ by 
\[
{\bar f}(x,y)=\frac{1}{n-1} \sum^{n-1}_{i=1}{\bar f}^{(i)}(x,y) \, ,
\]
for each $(x,y) \in [0,1]^2$. 

The crucial step in the proofs of Theorems~\ref{thm:two-layers:group-birth} and~\ref{thm:two-layers:group-birth:FPTAS} consists in establishing that $\omega = (\frac{1}{3},\frac{1}{3},\frac{1}{3})$ lies in the image of ${\bar f}$, and similarly for Theorem~\ref{thm:one-layer:group-birth:FPTAS}. We provide a brief description on how this is done, the full proof being given hereafter. Because the preferences are monotone, it is not difficult to see that the image of every basic line by ${\bar f}$ has its endpoint $y=0$ below the line 
$B=\{\, (z_1,z_2,z_3)  \mid z_1 = \frac 1 3, (z_1,z_2,z_3)  \in \Delta^2 \,\}$ and its endpoint $y=1$ above this line. Here, {\em above} means $z_1 \ge \frac 13$ and {\em below} means $z_1 \le \frac 13$. This implies in particular that the image of every basic line by ${\bar f}$ intersects with $B$ (by continuity). 
Further, we show that the images ${\bar f}(\bigL(0))$ and ${\bar f}(\bigL(1))$ of the sides enjoy symmetry with respect to the line $z_2 = z_3$. The existence of two adjacent vertical lines $\bigL(\frac{k}{N})$ and $\bigL(\frac{k+1}{N})$ enclosing a vertical strip whose image by ${\bar f}$ hits $\omega$ then follows. To identify these lines, both mathematically and algorithmically, we need to determine the relative position of the image of any basic line with respect to $\omega$. This is done with the help of its intersection with $L$ or $R$, where $L$ and $R$ are obtained by splitting $B$ by its middle point $\omega$, i.e.,  $L=\{\, (z_1,z_2,z_3)  \mid z_2  \le  \frac{1}{3}, (z_1,z_2,z_3)  \in B \,\}$ and $R=\{\, (z_1,z_2,z_3)  \mid z_2  \ge  \frac{1}{3}, (z_1,z_2,z_3)  \in B \,\}$.

We first establish a useful property of ${\bar f}$. In standard one-layered cake cutting into three pieces with monotone preferences, agents' preferences are monotone when moving one knife from left to right while fixing the other. In two-layered cake cutting, we have partial monotonicity, that is, agents' preferences exhibit this property when fixing the long knife. 
A map $g \colon [0,1]^2 \rightarrow \Delta^2$ is \emph{vertically monotone} with respect to $x \in [0,1]$ if the first coordinate of $g$ is nondecreasing when fixing $x$, i.e., $g_1(x,y') \ge g_1(x,y)$ for every pair $y,y' \in [0,1]$ with $y' \ge y$. 

\begin{lemma}\label{lem:monotone}
The map ${\bar f}$ is vertically monotone with respect to every $x=\frac k N$ for $k \in \{0,1,\ldots,N\}$. %
\end{lemma}

\begin{proof}
Fix $x=\frac{k}{N}$ for $k \in \{0,1,\ldots,N\}$. Consider any agent $i \in [n-1]$. 
On the basic line $\bigL(x)$, if a vertex is colored with $1$, all the vertices with bigger $y$ are colored with $1$ as well. 
This property is due to the tie-breaking rule and the monotonicity of preferences: 
as $y$ increases from $0$ to $1$, the piece $\mathcal{A}_1(x,y)$ gets bigger, while the second piece $\mathcal{A}_2(x,y)$ and the third piece $\mathcal{A}_3(x,y)$ do not grow. 
Thus, since each ${\bar f}^{(i)}$ is an affine extension of ${f}^{(i)}$, the function ${\bar f}^{(i)}$ is vertically monotone with respect to $x$. 
Further, since all the maps ${\bar f}^{(i)}$ are vertically monotone with respect to $x$, their average ${\bar f}$ also satisfies the same property.~
\end{proof}

Note that in general, this property may not hold for other $x \in [0,1]$.

The following lemma states that the images of $\bigL(0)$ and $\bigL(1)$ are symmetric with respect to the line $\{\, (z_1,z_2,z_3) \in \Delta^2 \mid z_2=z_3 \,\}$.

\begin{lemma}\label{lem:sym}
For every $y\in[0,1]$, the following holds:
\[
\bar f_1(0,y) = \bar f_1(1,y), \quad \bar f_2(0,y) = \bar f_3(1,y), \quad \bar f_3(0,y) = \bar f_2(1,y) \, .
\]
\end{lemma}

\begin{proof} Since the partitions induced by $(0,y)$ and $(1,y)$ are symmetric, we have: 
\begin{itemize}
\item if one of $(0,y)$ and $(1,y)$ is colored with $1$ by agent $i$ ({i.e., $f^{(i)}(0,y)=\e_1$ or $f^{(i)}(1,y)=\e_1$}), the other is also colored with $1$ by $i$. 
\item if one of $(0,y)$ and $(1,y)$ is colored with $2$ by agent $i$ ({i.e., $f^{(i)}(0,y)=\e_2$ or $f^{(i)}(1,y)=\e_2$}), the other is colored with $3$ by $i$. 
\end{itemize}
(When $0 \le y < 1$, the above symmetry holds because of the tie-breaking rule; in case of ties between $2$ and $3$, $2$ is chosen at $(0,y)$ since $y \ge 0$ while $3$ is chosen at $(1,y)$ since $y <1$. When $y=1$, it holds because $(0,1)$ receives color $1$ and $(1,1)$ receives color $1$.)
~
\end{proof}

To avoid unnecessary case distinctions, we consider a slightly perturbed version of $B$, which we denote by $B^{\delta}$ and which is defined by $B^{\delta}=\{\, (z_1,z_2,z_3)  \mid z_1 = \frac 1 3-\delta, (z_1,z_2,z_3)  \in \Delta^2 \,\}$.
We extend the previous notions attached to $B$. A point $(z_1,z_2,z_3)$ is \emph{above} (resp. \emph{below}) $B^{\delta}$ if $z_1 \ge \frac 1 3 - \delta$ (resp. $z_1 \le \frac 1 3 - \delta$). The point $\omega^{\delta}$ is the middle point of $B^{\delta}$ (with coordinates $(\frac 1 3-\delta,\frac 1 3+\frac \delta 2,\frac 1 3+ \frac \delta 2)$ thus), and $L^{\delta}$ and $R^{\delta}$ are defined by splitting $B^{\delta}$ by $\omega^{\delta}$, i.e.,  $L^{\delta}=\{\, (z_1,z_2,z_3)  \mid z_2  \le  \frac{1}{3} + \frac{\delta} 2, (z_1,z_2,z_3)  \in B^{\delta} \,\}$ and $R^{\delta}=\{\, (z_1,z_2,z_3)  \mid z_2  \ge  \frac{1}{3}+\frac{\delta} 2, (z_1,z_2,z_3)  \in B^{\delta} \,\}$.

Now, set $\delta \coloneqq \frac 1 {6n^2}$. The next two lemmas states that with this small $\delta$, while we get rid of some ``degeneracy'' (images of vertices located on $B$), we can safely replace $\omega$ by $\omega^{\delta}$ in the task of finding a triangle of $\T$ containing $\omega$ in its image by $\bar f$.

\begin{lemma}\label{lem:perturb-rel}
No vertex of $\T$ is mapped by $\bar f$ on $B^{\delta}$. Moreover, if a vertex is mapped by $\bar f$ above (resp. below) $B^{\delta}$, it is mapped above (resp. below) $B$.
\end{lemma}

\begin{proof}
Consider a vertex $\vv$ of $\T$. By definition of $\bar f$, there is $\ell_1 \in \{0,1,\ldots,n-1\}$ such that $\bar f_1(\vv) = \frac {\ell_1} {n-1}$. The quantities $\bar f_1(\vv) - (\frac 1 3 - \delta)$ and $6\ell_1 - \big(2(n-1)-\frac {n-1} {n^2}\big)$ are equal up to a multiplication by $6(n-1)$ and are thus either both positive, or both zero, or both negative. Since $n \ge 2$, the quantity $\frac {n-1} {n^2}$ is not an integer, which implies that $\bar f_1(\vv) - (\frac 1 3 - \delta)$ is not $0$. This already proves the first part of the statement.

To prove the second part, assume first that $\bar f(\vv)$ is above $B^{\delta}$, namely that $\bar f_1(\vv) \ge \frac 1 3 - \delta$. The same computation as above shows then that $6\ell_1 \ge 2(n-1)-\frac {n-1} {n^2}$. Since $n-1<n^2$, we have $6\ell_1 \ge 2(n-1)$. This latter inequality translates into $\bar f_1(\vv) \ge \frac 1 3$. The case $\bar f(\vv)$ below $B^{\delta}$ is immediate: the inequality $\bar f_1(\vv) \le \frac 1 3 - \delta$ implies $\bar f_1(\vv) \le \frac 1 3$.
\end{proof}

\begin{lemma}\label{lem:perturb-tri}
For every triangle $\tau$ of $\T$, if $\omega^{\delta}$ belongs to $\bar f(\tau)$, then so does $\omega$.
\end{lemma}

\begin{proof} Consider two vertices $\vv_1$ and $\vv_2$ of $\T$. The sign of the determinant
\[
\varphi(z_1,z_2,z_3) = 
\left|\begin{array}{ccc}
z_1 & z_2 & z_3
\\[0.6ex] \bar f_1(\vv_1) & \bar f_2(\vv_1) & \bar f_3(\vv_1) \\[1ex] \bar f_1(\vv_2) & \bar f_2(\vv_2) & \bar f_3(\vv_2)
\end{array}
\right|
\]
records the relative position of a point $\z=(z_1,z_2,z_3) \in \Delta^2$ with respect to the line going through $\bar f(\vv_1)$ and $\bar f(\vv_2)$. Since $\bar f_j(\vv)$ is of the form $\frac {\ell_j} {n-1}$ with $\ell_j \in \{0,1,\ldots,n-1\}$ for every $j \in \{1,2,3\}$ and every vertex $\vv$ of $\T$, there are three integers $a,b,c$ smaller than $n^2$ in absolute value such that
\[
\varphi(\omega) = \frac 1 {3(n-1)^2} (a+b+c)
\quad\text{and}\quad
\varphi(\omega^{\delta}) = 
\frac 1 {3(n-1)^2}\left((a+b+c)-\frac 1 {4n^2}(2a-b-c)\right) \, .
\] 
Since each of $a$, $b$, and $c$ is smaller than $n^2$ in absolute value, we have $|2a-b-c|<4n^2$. In case $\varphi(\omega^{\delta})=0$, then this prevents $a+b+c$ to be different from $0$ (it is an integer number), and we have $\varphi(\omega)=0$ as well. In case $\varphi(\omega^{\delta})>0$, then this prevents $a+b+c \le -1$ to hold, and we have $\varphi(\omega) \ge 0$. In case $\varphi(\omega^{\delta})<0$, then this prevents $a+b+c \ge 1$ to hold, and we have $\varphi(\omega) \le 0$.

Now, take a triangle $\tau$ of $\T$ such that $\omega^{\delta} \in \bar f(\tau)$. No vertex of $\tau$ is mapped on $\omega^{\delta}$ by $\bar f$: otherwise, this point would be mapped on $B^{\delta}$, which is not possible by Lemma~\ref{lem:perturb-rel}. So, there are two possibilities: either $\tau$ has an edge whose image by $\bar f$ is a nondegenerate segment containing $\omega^{\delta}$ in its relative interior; or the image of $\tau$ by $\bar f$ is a nondegenerate triangle containing $\omega^{\delta}$ in its relative interior. 

Suppose that the first possibility occurs: the point $\omega^{\delta}$ lies in the relative interior of $[\bar f(\vv_1),\bar f(\vv_2)]$ for some vertices $\vv_1$ and $\vv_2$ (with distinct images by $\bar f$). It means $\varphi(\omega^{\delta})=0$. We have noticed that $\varphi(\omega)=0$ holds then, which means that $\omega$ lies on the line supported by the segment $[\bar f(\vv_1),\bar f(\vv_2)]$. Lemma~\ref{lem:perturb-rel} then shows that the relative positions of $\omega^{\delta}$ and $\omega$ with respect to $\bar f(\vv_1)$ and $\bar f(\vv_2)$ are the same, which means that $\omega$ is also in $[\bar f(\vv_1),\bar f(\vv_2)]$ (but not necessarily in its relative interior).

Suppose then that the second possibility occurs: the point $\omega^{\delta}$ lies in the relative interior of $\operatorname{conv}(\bar f(\vv_1),\bar f(\vv_2),\bar f(\vv_3))$ for some vertices $\vv_1$, $\vv_2$, and $\vv_3$ (with pairwise distinct images by $\bar f$). For each pair of vertices $\vv_i,\vv_j$ with $i\neq j$, the determinant $\varphi(\omega^{\delta})$ is either positive or negative (but not $0$); we have noticed that in the first case $\varphi(\omega) \ge 0$ holds and that in the second case $\varphi(\omega) \le 0$ holds. This means that, for every edge $e$ of $\tau$, the points $\omega$ and $\omega^{\delta}$ have the same relative position with respect to the line containing $\bar f(e)$. This implies that $\omega$ lies in the image of $\tau$ by $\bar f$ (but not necessarily in its relative interior).
\end{proof}

The next lemma will be used to ensure the existence of adjacent basic lines forming a strip whose image by $\bar f$ hits $\omega^{\delta}$. Uniqueness will be crucial for the algorithmic determination of the relative position of the image of a basic line with respect to $\omega^{\delta}$. Note that this relative position is completely determined by whether $\bar f(x,y)$ lies on $L^{\delta}$ or on $R^{\delta}$.

\begin{lemma}\label{lem:usquare:odd}
For every $x=\frac k N$ with $k \in \{0,1,\ldots,N\}$, there exists $y \in [0,1]$ such that $\bar f(x,y)$ lies on $B^{\delta}$, and this $y$ is unique.
\end{lemma}

\begin{proof}
We first establish that ${\bar f}(x,0)$ lies ``below'' $B^{\delta}$ and that ${\bar f}(x,1)$ lies ``above'' $B^{\delta}$. To see that ${\bar f}(x,0)$ lies ``below'' $B^{\delta}$, i.e., that ${\bar f}_1(x,0) \le \frac 1 3 - \delta$, note that the piece $\mathcal{A}_1(x,0)$ cannot be strictly preferred on the boundary of $y=0$ since it is of zero-length. Thus, by the tie-breaking rule, we have $f^{(i)}(x,0)=\e_2$ or $f^{(i)}(x,0)=\e_3$ for each agent $i \in [n-1]$, meaning that ${\bar f}_1(x,0)=\frac{1}{n-1}\sum^{n-1}_{i=1}f^{(i)}_1(x,0)=0$.

We show now that ${\bar f}(x,1)$ lies ``above'' $B^{\delta}$, i.e., that ${\bar f}_1(x,1) \ge \frac 1 3 - \delta$. Let $X$ be the set of agents whose preferred piece at vertex $(0,1)$ is $1$, i.e., $X=\{\, i \in [n-1] \mid f^{(i)}(0,1)=\e_1 \,\}$. Observe that $|X| \ge \lceil{\frac{n}{2}}\rceil$ by the tie-breaking rule and by the assumption that the majority of agents weakly prefer the top layer to the bottom layer.
By monotonicity of preferences, every agent in $X$ answers $\mathcal{A}_1(x,1)$ as a preferred piece. Hence,
\[
{\bar f}_1(x,1) = \frac{1}{n-1}\sum_{i \in X}f^{(i)}_1(x,1) = \frac{|X|}{n-1} \ge \frac{\lceil{\frac{n}{2}}\rceil}{n-1} > \frac{1}{3} - \delta \, .
\]

By Lemmas~\ref{lem:monotone} and~\ref{lem:perturb-rel}, there is a $k' \in \{0,1,\ldots, N-1\}$ such that $\bar f_1(x,\frac k N) < \frac 1 3 - \delta$ for every $k \in \{0,1,\ldots,k'\}$ and such that $\bar f_1(x,\frac k N) > \frac 1 3 - \delta$ for every $k \in \{k'+1,k'+2,\ldots,N\}$. There exists thus a single $y$ as in the statement, determined by the intersection of the segment $[\bar f(x,\frac {k'} N),\bar f(x,\frac {k'+1} N)]$ with the line $B^{\delta}$.
\end{proof}

Next, we show the relation between the line $B^{\delta}$ and the image of the boundary of each basic square. It is a well-known fact from topology that a sufficiently regular closed curve, as is the image of the boundary of a basic square by ${\bar f}$, intersects with any generic straight line an even number of times. The following lemma shows that we can be more precise in our case. We provide a direct proof. An alternate proof based on the mentioned fact from topology would also be possible and short.

\begin{lemma}\label{lem:square}
For a basic square $S$, either the image of the boundary by ${\bar f}$ does not intersect with $B^{\delta}$, or there are exactly two points of $S$, located on two distinct edges, whose image by ${\bar f}$ lies on $B^{\delta}$.
\end{lemma}

\begin{proof}
Let $\boldsymbol{a},\boldsymbol{b},\boldsymbol{c},\boldsymbol{d}$ be the top left, top right, bottom right, bottom left corner vertices of $S$, respectively. (We follow a clockwise order.) 
Recall that ${\bar f}$ is affine, so the image of $S$'s boundary by ${\bar f}$ is given by the union of $[{\bar f}(\boldsymbol{a}),{\bar f}(\boldsymbol{b})]$, $[{\bar f}(\boldsymbol{b}),{\bar f}(\boldsymbol{c})]$, $[{\bar f}(\boldsymbol{c}),{\bar f}(\boldsymbol{d})]$, and $[{\bar f}(\boldsymbol{d}),{\bar f}(\boldsymbol{a})]$.
Any of these segments intersects with $B^{\delta}$ if and only if one endpoint is above $B^{\delta}$ and the other is below $B^{\delta}$. None of these endpoints are located on $B^\delta$ by the choice of $\delta$ (Lemma~\ref{lem:perturb-rel}). There is thus an even number of edges of $S$ whose image by ${\bar f}$ intersects with $B^{\delta}$ (just because such an edge corresponds to going from one side to the other of $B^\delta$).

Suppose for a contradiction this number is equal to four. It means that each segment is intersecting with $B^\delta$, which implies that ${\bar f}(\a)$ and ${\bar f}(\cc)$ are on one side of $B^\delta$, and ${\bar f}(\boldsymbol{b})$ and ${\bar f}(\boldsymbol{d})$ are on the other side. Thus, if ${\bar f}(\a)$ and ${\bar f}(\cc)$ are above $B^\delta$, then ${\bar f}(\boldsymbol{b})$ and ${\bar f}(\boldsymbol{d})$ are below $B^\delta$; and if ${\bar f}(\a)$ and ${\bar f}(\cc)$ are below $B^\delta$, then ${\bar f}(\boldsymbol{b})$ and ${\bar f}(\boldsymbol{d})$ are above $B^\delta$. In any case, this is a contradiction because of ${\bar f}$ being vertically monotone (see Lemma~\ref{lem:monotone}).  \end{proof}

The following lemma ensures that if the image of the boundary of a basic square intersects with both $L^{\delta}$ and $R^{\delta}$, it contains a point whose image by ${\bar f}$ is $\omega^{\delta}$. 

\begin{lemma}\label{lem:square:LR}
Consider a basic square whose boundary has its image by ${\bar f}$ intersecting with $L^{\delta}$ and with $R^{\delta}$. Then the image of the basic square by ${\bar f}$ contains $\omega^{\delta}$.
\end{lemma}
\begin{proof}
Consider a basic square $S$. Let $\boldsymbol{a},\boldsymbol{b},\boldsymbol{c},\boldsymbol{d}$ be the top left, top right, bottom right, bottom left corner vertices of $S$, respectively. By Lemma~\ref{lem:perturb-rel}, none of the points $\bar f(\boldsymbol{a})$, $\bar f(\boldsymbol{b})$, $\bar f(\boldsymbol{c})$, $\bar f(\boldsymbol{d})$ are located on $B^{\delta}$. 

Suppose first that ${\bar f}(\a)$ and ${\bar f}(\cc)$ both lie above $B^\delta$. Since ${\bar f}$ is vertically monotone (see Lemma~\ref{lem:monotone}), ${\bar f}(\boldsymbol{b})$ also lies above $B^\delta$ and only the segments with ${\bar f}(\boldsymbol{d})$ as an endpoint can intersect with $B^\delta$. Thus ${\bar f}(\boldsymbol{d})$ lies below $B^\delta$. The two points in $\partial S$ whose image by $\bar f$ belongs to $B^\delta$ (see Lemma~\ref{lem:square}) are located on $[\cc,\boldsymbol{d}]$ and $[\boldsymbol{d},\a]$. By assumption, one image is located on $L^{\delta}$ and the other on $R^{\delta}$. Since $\bar f$ is piecewise affine, there is a point $\x$ in the triangle of $\T$ with vertices $\a,\cc,\boldsymbol{d}$ such that ${\bar f}(\x)=\omega^{\delta}$. When ${\bar f}(\a)$ and ${\bar f}(\cc)$ both lie below $B^\delta$, we get the result similarly.

Finally, suppose that ${\bar f}(\a)$ and ${\bar f}(\cc)$ do not lie on the same side of $B^\delta$. It implies that the segment $[{\bar f}(\a),{\bar f}(\cc)]$ intersects with $B^\delta$. If it contains the point $\omega^{\delta}$, we are done. We can thus assume that the intersection of $[{\bar f}(\a),{\bar f}(\cc)]$ and $B^\delta$ is distinct from this point. Consider the case when the intersection is on $R^\delta$. Let $\tau$ be the triangle of $\T$ containing the edge of $S$ whose image intersects with $L^\delta$. This triangle has also $[\a,\cc]$ as an edge.
Since $\bar f$ is piecewise affine, there is a point $\x$ in $\tau$ such that ${\bar f}(\x)=\omega^{\delta}$. The case when the intersection is on $L^\delta$ is dealt with similarly.~
\end{proof}

See Figure~\ref{fig:intersection} for examples of the image of a basic square by $\bar f$. 
\begin{figure}[htb]
\centering
\begin{tikzpicture}[scale=0.5, transform shape]
       \filldraw[black!30] (-2,1) -- (-4,2.2) -- (0,2) -- (0,-2);
       
       \draw[ultra thick] (-6,-0.5) -- (2,-0.5);
       \node at (-6.8,-0.5) {\Large $B^{\delta}$};
       \draw[fill=black] (-2,1) circle[radius=0.2];
       \draw[fill=black] (-4,2.2) circle[radius=0.2];
       \draw[fill=black] (0,2) circle[radius=0.2];
       \draw[fill=black] (0,-2) circle[radius=0.2];
       
       \node at (-2.8,0.8) {\Large ${\bar f}(\boldsymbol{a})$};
       \node at (0.8,-2.2) {\Large ${\bar f}(\boldsymbol{d})$};
       \node at (-4.8,2.2) {\Large ${\bar f}(\boldsymbol{b})$};
       \node at (0.8,2) {\Large ${\bar f}(\boldsymbol{c})$};
        
    \begin{scope}[xshift=8cm,yshift=-1cm]
        \filldraw[black!30] (2,3) -- (1,1) -- (2,-1) -- (-2,-0.5);
        \draw[ultra thick] (-4,0.5) -- (4,0.5);
        \node at (-4.8,0.5) {\Large $B^{\delta}$};
        
        \draw[fill=black] (2,3) circle[radius=0.2];
        \draw[fill=black] (1,1) circle[radius=0.2];
        \draw[fill=black] (2,-1) circle[radius=0.2];
        \draw[fill=black] (-2,-0.5) circle[radius=0.2];
        
        \node at (2.9,3) {\Large ${\bar f}(\boldsymbol{b})$};
        \node at (2,1) {\Large ${\bar f}(\boldsymbol{c})$};
        \node at (3,-1) {\Large ${\bar f}(\boldsymbol{d})$};
        \node at (-3,-0.5) {\Large ${\bar f}(\boldsymbol{a})$};
        
    \end{scope}
    
    \begin{scope}[xshift=20cm,yshift=0cm]
       \filldraw[black!30] (-4,-1.5) -- (-4,1) -- (-1,1) -- (-1,-1.5);
       \draw[ultra thick] (-6,-0.5) -- (1,-0.5);
       \node at (-6.8,-0.5) {\Large $B^{\delta}$};
       
    \draw[fill=black] (-4,-1.5) circle[radius=0.2];
      \draw[fill=black] (-4,1) circle[radius=0.2];
      \draw[fill=black] (-1,-1.5) circle[radius=0.2];
      \draw[fill=black] (-1,1) circle[radius=0.2];
       \node at (-4.8,1) {\Large ${\bar f}(\boldsymbol{a})$};
       \node at (-4.8,-1.5) {\Large ${\bar f}(\boldsymbol{d})$};
       \node at (-0.2,1) {\Large ${\bar f}(\boldsymbol{b})$};
       \node at (-0.1,-1.5) {\Large ${\bar f}(\boldsymbol{c})$};
    \end{scope}
\end{tikzpicture}
\caption{Illustration of the image of a basic square by ${\bar f}$.} 
\label{fig:intersection}
\end{figure}
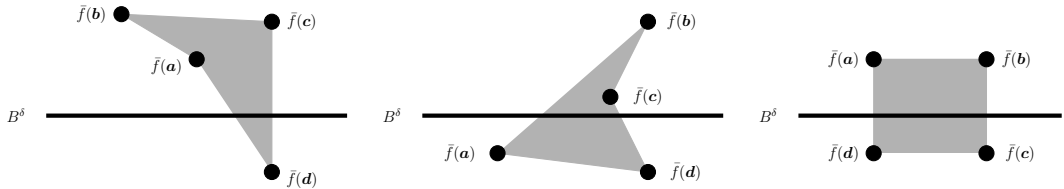

\smallskip

\subsection{Existence of an envy-free multi-division}\label{subsec:two:existence}

In this section, we prove Theorem~\ref{thm:two-layers:group-birth}. 
The main lemma is the following one. Figure~\ref{fig:average} illustrates its proof.

\begin{lemma}\label{lem:average}
There is a point of $[0,1]^2$ whose image by $\bar f$ is $\omega^{\delta}$.
\end{lemma}
\begin{proof}
By Lemmas~\ref{lem:sym} and~\ref{lem:usquare:odd}, ${\bar f}(\bigL(0))\cup{\bar f}(\bigL(1))$ intersects with $L^{\delta}$ and with $R^{\delta}$. Further, by Lemma~\ref{lem:usquare:odd}, for each $x=\frac k N$ with $k \in \{1,2, \ldots,N\}$, the line ${\bar f}(\bigL(x))$ intersects with $B^{\delta}$ at a unique point. 
Thus, there are two adjacent basic lines $\bigL(x_1)$ and $\bigL(x_2)$ with $x_2= x_1+\frac 1 N$ or $x_2=x_1-\frac 1 N$ such that ${\bar f}(\bigL(x_1))$ intersects with $L^{\delta}$ and ${\bar f}(\bigL(x_2))$ intersects and with $R^{\delta}$.

Let $e_1$ be the vertical edge of $\bigL(x_1)$ such that ${\bar f}(e_1)$ intersects with $L^{\delta}$ and $e_2$ the vertical edge of $\bigL(x_2)$ such that ${\bar f}(e_2)$ intersects with $R^{\delta}$. Let $S_1$ be the basic square that lies between $\bigL(x_1)$ and $\bigL(x_2)$ and contains $e_1$ and $S_2$ the basic square that lies between $\bigL(x_1)$ and $\bigL(x_2)$ and contains $e_2$. 
By Lemma~\ref{lem:square}, for $i=1,2$, $|\{\x\in\partial S_i \mid  {\bar f}(\x) \in B^\delta\}|=2$. If the image of the boundary of $S_1$ by ${\bar f}$ intersects with $L^{\delta}$ and $R^{\delta}$, Lemma~\ref{lem:square:LR} implies that the  image of $S_1$ by ${\bar f}$ contains $\omega^{\delta}$, which proves the claim.  Similarly, if the image of the boundary of $S_2$ by ${\bar f}$ intersects with $L^{\delta}$ and $R^{\delta}$, we obtain the desired claim.   

Thus, assume that the image of $S_1$ by ${\bar f}$ intersects with $L^{\delta}$ only, and the image of $S_2$ by ${\bar f}$ intersects with $R^{\delta}$ only. By Lemma~\ref{lem:square}, this means that there are two edges in $S_1$ whose image by ${\bar f}$ intersects with $L^{\delta}$. Likewise, there are two edges in $S_2$ whose image by ${\bar f}$ intersects with $R^{\delta}$. Consider the vertical rectangle $C$ of width $\frac 1 N$ that consists of the sequence of basic squares lying between $S_1$ and $S_2$. We claim that $C$ admits two horizontal edges of a same basic square whose image by ${\bar f}$ intersects with both $L^{\delta}$ and $R^{\delta}$ and hence contains $\omega^{\delta}$ by Lemma~\ref{lem:square:LR}. 

To prove this, recall that by Lemma~\ref{lem:usquare:odd}, the image of $\bigL(x_1)$ by $\bar f$ intersects with $B^{\delta}$ at a single point. Thus, there is no vertical edge in $\bigL(x_1)$ except for $e_1$ whose image by ${\bar f}$ intersects with $B^{\delta}$. Similarly, there is no vertical edge in $\bigL(x_2)$ except for $e_2$ whose image by ${\bar f}$ intersects with $B^{\delta}$. Further, the image of the top horizontal edge of $C$ does not intersect with $B^{\delta}$ since ${\bar f}$ is affine on the edge and the corner vertices of the edge have their images lying above $B^{\delta}$. For the same reasons, the image of the bottom horizontal edge of $C$ does not intersect with $B^{\delta}$. 

By Lemma~\ref{lem:square} and by the fact that the image of each $S_i$ $(i=1,2)$ by ${\bar f}$ intersects with $B^{\delta}$, all squares in $C$ have two edges whose image intersect with $B^{\delta}$. 
Hence, the image of every horizontal edge in $C$ except for the top and bottom ones by ${\bar f}$ intersects with $L^{\delta}$ or $R^{\delta}$. This means that when moving from $S_1$ to $S_2$, one necessarily meets a basic square whose horizontal edges have their image by ${\bar f}$ intersecting both $L^{\delta}$ and $R^{\delta}$. 
Thus, by Lemma~\ref{lem:square:LR}, we have found a basic square $S^*$ whose image by ${\bar f}$ contains $\omega^{\delta}$.~
\end{proof}

\begin{figure}[htb]
\centering
    \begin{tikzpicture}[scale=0.4, transform shape]

    \draw [<->,thick] (-1,13) node (yaxis) [above] {\Large $y$}|- (13.5,-1) node (xaxis) [right] {\Large $x$};
    
    \draw[fill=black!20] (4,10) rectangle (6,2);
    
    \foreach \x in {0,2,4,6, 8,10,12}{
    \foreach \y in {0,2,4,6, 8,10,12}{
            \draw[color=gray!50,ultra thick] (0,0) rectangle (\x,\y);
        }
    }
    
    \draw[thick] (-1.2,0) node [left] {\Large $0$} -- (-0.8,0);
    \draw[thick] (-1.2,12) node [left] {\Large $1$} -- (-0.8,12);

    \draw[thick] (0,-1.2) node [below] {\Large $0$} -- (0,-0.8);
    \draw[thick] (12,-1.2) node [below] {\Large $1$} -- (12,-0.8);
    
    \draw[thick] (0,-1.2) node [below] {\Large$0$} -- (0,-0.8);
    
    \draw[thick] (4,-1.2) node [below] {\Large $x_2$} -- (4,-0.8);
    \draw[thick] (6,-1.2) node [below] {\Large $x_1$} -- (6,-0.8);
    
    \node at (5,9) {\Large $S_2$};
    \node at (3.5,9) {\Large $e_2$};
    \node at (5,3) {\Large $S_1$};
    \node at (6.5,3) {\Large $e_1$};
    
    \draw[color=blue,ultra thick] (4,10) rectangle (4,8);
    \draw[color=blue,ultra thick] (4,8) rectangle (6,8);
    \draw[color=blue,ultra thick] (4,6) rectangle (6,6);
    \draw[color=red,ultra thick] (4,4) rectangle (6,4);
    \draw[color=red,ultra thick] (6,2) rectangle (6,4);

    \node at (5,5) {\Large $S^*$};
    
\begin{scope}[xshift=20cm,yshift=-1cm]
    \draw[fill=black] (0,0) circle (3pt);
    \draw[fill=black] (16,0) circle (3pt);
    \draw[fill=black] (8,13.84) circle (3pt);
    \draw[thick] (0,0) -- (16,0) -- (8,13.84) -- (0,0);
    
    \filldraw[gray!50] (5,4.9) -- (7.9,2.5) -- (9,5) -- (8.5,6);
    
    \draw[ultra thick] (2.1,3.7) -- (13.9,3.7);
    \node at (1.3,3.7) {\Large $B^{\delta}$};
    
    \draw[thick,gray] (2.67,4.61) -- (13.3,4.61);
    \node at (2,4.61) {\Large $B$};
    
    \draw[fill=black] (8,4.61) circle (5pt);
    
    \draw[thick,dotted] (2,0)--(3,2) (5,4.9)--(8.5,6)--(9.5,7)--(7,8) -- (6,9);
    
    \node at (2,-0.5) {\Large ${\bar f}(\bigL(x_1))$};
    
    \draw[thick,dotted] (10,0)-- (8,1)-- (7.9,2.5)--(9,5)--(9.5,5.2) (10.3,8.5)--(9,9);
    \node at (10,-0.5) {\Large ${\bar f}(\bigL(x_2))$};
    
    \draw[ultra thick,red] (3,2)--(5,4.9)--(7.9,2.5); 
    
    \draw[ultra thick,blue] (10.3,8.5)--(9.5,5.6) -- (9.5,7) (9,5) -- (8.5,6);
    
    \draw[fill=black] (2,0) circle (3pt);
    \draw[fill=black] (3,2) circle (3pt);
    \draw[fill=black] (5,4.9) circle (3pt);
    \draw[fill=black] (8.5,6) circle (3pt);
    \draw[fill=black] (9.5,7) circle (3pt);
    \draw[fill=black] (7,8) circle (3pt);
    \draw[fill=black] (6,9) circle (3pt);
    
    \draw[fill=black] (10,0) circle (3pt);
    \draw[fill=black] (8,1) circle (3pt);
    \draw[fill=black] (7.9,2.5) circle (3pt);
    \draw[fill=black] (9,5) circle (3pt);
    \draw[fill=black] (9.5,5.2) circle (3pt);
    \draw[fill=black] (10.3,8.5) circle (3pt);
    \draw[fill=black] (9,9) circle (3pt);
    
\end{scope}
\end{tikzpicture}
\caption{Illustration of the proof of Lemma~\ref{lem:average}. The left figure depicts the unit square where the gray region corresponds to the rectangle $C$; the blue edges correspond to the edges whose image by ${\bar f}$ intersects with $R^{\delta}$; and the red edges correspond to the edges whose image by ${\bar f}$ intersects with $L^{\delta}$. The right figure depicts the image by ${\bar f}$ of basic lines $\bigL(x_1)$ and $\bigL(x_2)$ as well as the colored edges of the left figure on the standard simplex $\Delta^2$.}
\label{fig:average}
\end{figure}
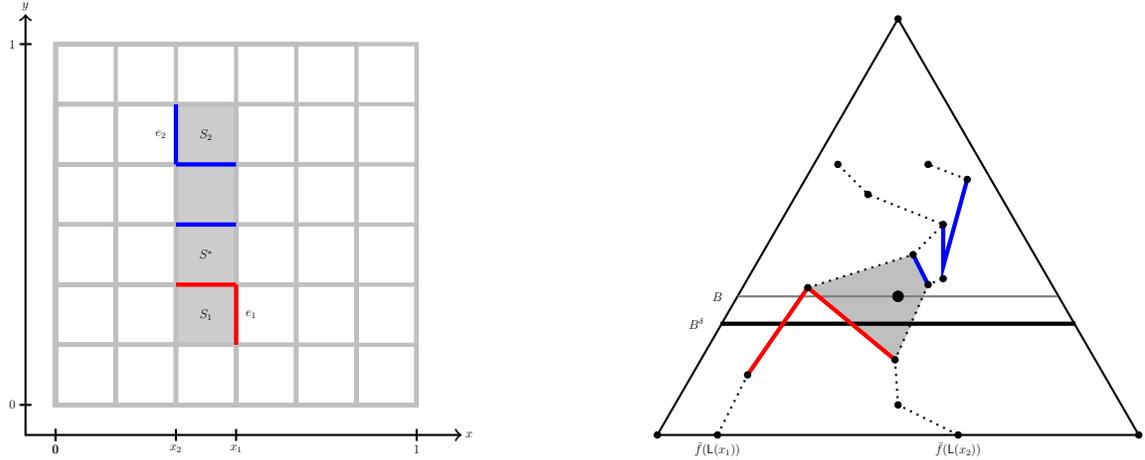%

\smallskip
\begin{proof}[Proof of Theorem~\ref{thm:two-layers:group-birth}.] The proof is almost verbatim the end of the proof of Theorem~\ref{thm:group-birth}. 

By Lemma~\ref{lem:average}, there is a point of $[0,1]^2$ whose image by $\bar f$ is $\omega^{\delta}$. According to Lemma~\ref{lem:perturb-tri}, this implies the existence of a point $(x^*,y^*) \in [0,1]^2$ such that $\bar f(x^*,y^*) = \omega$. For every non-birthday agent $i \in [n-1]$ and every layered piece $j \in \{1,2,3\}$, we define $w_{ij} = {\bar f}^{(i)}_j(x^*,y^*)$. By definition of ${\bar f}$ and the fact that ${\bar f}_j(x^*,y^*) = \frac 13$ for all $j=1,2,3$, we have 
$$
\sum_{i=1}^{n-1} w_{ij} = \sum_{i=1}^{n-1}{\bar f}^{(i)}_j(x^*,y^*)=(n-1){\bar f}_j(x^*,y^*) = \frac {n-1} 3
$$
for all $j=1,2,3$. Further, 
$$
\sum_{j=1}^3 w_{ij}=\sum_{j=1}^3 {\bar f}^{(i)}_j(x^*,y^*) = 1
$$
for all $i \in [n-1]$ because ${\bar f}^{(i)}$ has its image in $\Delta^2$. 
Consider the bipartite graph $H=([n-1],\{1,2,3\};E)$ where the edge $ij$ exists precisely when $w_{ij}>0$. Applying Lemma~\ref{lem:TUM} to $H$ with $a_j= \frac{n-1}{3}$, for every $j^* \in \{1,2,3\}$, there exists an assignment $\pi_{j^*} \colon [n] \rightarrow \{1,2,3\}$ such that 
\begin{itemize}
    \item $\pi_{j^*}(n)=j^*$, 
    \item for each $i \in [n-1]$, the vertex $\pi_{j^*}(i)$ is a neighbor of $i$ in $H$, 
    \item $|\pi^{-1}_{j^*}(n)| = \lfloor (n-1)/3 \rfloor +1$, and
    \item for each $j \in \{1,2,3\}$, we have $|\pi^{-1}_{j^*}(j)| \in \{ \lfloor (n-1)/3 \rfloor, \lceil (n-1)/3 \rceil \}$.
\end{itemize}
Similar to the proof of Theorem~\ref{thm:group-birth}, we have $|\pi^{-1}_{j^*}(j)|\in\{\lfloor n/3\rfloor, \lceil n/3\rceil \}$ for all $j\in \{ 1,2,3\}$. The rest of the proof is the same as the last two paragraphs in the proof of Theorem~\ref{thm:group-birth}. (Especially, we make $N$ go to infinity.)~ 
\end{proof}
  
We remark that it is unclear whether the statement of Theorem~\ref{thm:two-layers:group-birth} remains true if we consider more general choice functions. The problem appears from the fact that otherwise the top side $\{(x,1)\mid x\in[0,1]\}$ of the unit square could also receive colors $2$ or $3$.

\subsection{FPTAS for a two-layered cake}\label{sec:two:FPTAS}
In this section, we prove Theorem~\ref{thm:two-layers:group-birth:FPTAS}. Instead of having the general preferences as in Section~\ref{subsec:two:existence}, we have now valuation functions that satisfy monotonicity and the Lipschitz condition with constant $K$. In a similar spirit to the work of Deng et al. \cite{Deng2012}, we employ a divide-and-conquer approach to compute an approximate solution. Deng et al.\ use a triangulation of the standard triangle to encode possible divisions, compute an approximate solution by recursively computing smaller polygons containing such a solution, and finish when the polygon is actually a small triangle of the triangulation. We work instead on the unit square subdivided into basic squares, recursively compute thinner full-height rectangles containing an approximate solution, and end up with a very thin vertical strip in which we identify a basic square containing an approximate solution. In both approaches, the part which is kept at any iteration is identified by computing a sort of ``degree'' of the boundary, which is related to the intersections with $L^{\delta}$ and $R^{\delta}$. The uniqueness ensured by Lemma~\ref{lem:usquare:odd} allows to completely determine these intersections via binary search.


Take any $\varepsilon>0$. We define $N$ to be $\lceil \frac {6K}{\varepsilon} \rceil$ and consider the triangulation $\T$ and the map ${\bar f}$ for this $N$ (their definitions depend on $N$; see Section~\ref{subsec:preliminaries}). 


We first show that, by the the Lipschitz condition of the valuation functions, any vertex of a basic square whose image by $\bar f$ contains $\omega$ translates into a desired approximate multi-division.


\begin{lemma}\label{lem:approximate-basicrectangle}
Every vertex $\vv$ of a basic square whose image by $\bar f$ contains $\omega$ encodes a multi-division $\mathcal{A}(\vv)$ such that for every $j^*\in\{1,2,3\}$, there is an assignment $\pi_{j^*}\colon[n]\rightarrow\{1,2,3\}$ with:
\begin{itemize}
    \item $\pi_{j^*}(n)=j^*$,
    \item for each $i\in [n-1]$, $v_i(\calA_{\pi_{j^*}(i)}(\vv))+\varepsilon \ge \max_{j \in \{1,2,3\}}v_i(\calA_{j}(\vv))$, and
    \item for each $j\in \{1,2,3\}$, $|\pi^{-1}_{j^*}(j)|\in\{\lfloor n/3\rfloor, \lceil n/3\rceil \}$.
\end{itemize} 
\end{lemma}
\begin{proof} Consider a basic square containing a point $(x^*,y^*)$ such that $\bar f(x^*,y^*) = \omega$, and let $\vv$ be any of its vertices. We define $w_{ij} = {\bar f}^{(i)}_j(x^*,y^*)$ for every non-birthday agent $i \in [n-1]$ and every layered piece $j \in \{1,2,3\}$. Similar to the proof of Theorem~\ref{thm:two-layers:group-birth}, we apply Lemma~\ref{lem:TUM} to the bipartite graph $H=([n-1],\{1,2,3\};E)$, where the edge $ij$ exists precisely when $w_{ij}>0$ and $a_j=\frac{n-1}{3}$ for $j \in \{1,2,3\}$. Then, for every $j^* \in \{1,2,3\}$, we get an assignment $\pi_{j^*} \colon [n] \rightarrow \{1,2,3\}$ such that 
\begin{itemize}
    \item $\pi_{j^*}(n)=j^*$,
    \item for each $i \in [n-1]$,  the vertex $\pi_{j^*}(i)$ is a neighbor of $i$ in $H$, and
    \item for each $j \in \{1,2,3\}$, we have $|\pi^{-1}_{j^*}(j)| \in \{\lfloor n/3\rfloor, \lceil n/3\rceil \}$. 
\end{itemize}
Consider any $j^* \in \{1,2,3\}$ and any $i \in [n-1]$. 
Since $\pi_{j^*}(i)$ is a neighbor of $i$ in $H$, there exists a vertex $\vv^{i,j^*}$ of the supporting simplex of $(x^*,y^*)$ in $\T$ such that ${f}^{(i)}(\vv^{i,j^*})=\e_{\pi_{j^*}(i)}$, meaning that $v_i(\calA_{\pi_{j^*}(i)}(\vv^{i,j^*})) \ge \max_{j \in \{1,2,3\}}v_i(\calA_{j}(\vv^{i,j^*}))$. By the Lipschitz condition, we have for every $j$:
\begin{align*}\label{eq:approximate}
|v_i(\mathcal{A}_j(\vv))-v_i(\mathcal{A}_j(\vv^{i,j^*}))| \le K \times \mu(\mathcal{A}_j(\vv) \triangle \mathcal{A}_j(\vv^{i,j^*})) \le K \times \frac{3}{N} \le \frac{\varepsilon} 2 \, . 
\end{align*}
(Here, the `$3$' in the penultimate term comes from the following: consider two vertices of a same basic square; the layered pieces $j=1$ in the multi-divisions they encode have a symmetric difference of length $0$ or $\frac 1 N$; the layered pieces $j=2$ in the multi-divisions they encode have a symmetric difference of length $0$, $\frac 1 N$, or $\frac 2 N$; same thing for the layered pieces $j=3$.)
We have thus:
\begin{itemize}
    \item $\pi_{j^*}(n)=j^*$,
    \item for each $i\in [n-1]$, $v_i(\calA_{\pi_{j^*}(i)}(\vv))+\varepsilon \ge \max_{j \in \{1,2,3\}}v_i(\calA_{j}(\vv))$, and
    \item for each $j\in \{1,2,3\}$, we have $|\pi^{-1}_{j^*}(j)|\in\{\lfloor n/3\rfloor, \lceil n/3\rceil \}$.~\qedhere
\end{itemize} 
\end{proof}

The problem thus boils down to finding a basic square whose image by ${\bar f}$ includes $\omega$ in logarithmic time of $N$. In order to design the FPTAS, we exploit the uniqueness of Lemma~\ref{lem:usquare:odd} (which relies on Lemma~\ref{lem:monotone}, which eventually relies on the preferences being monotone). This enables us to compute the intersection of the image of a basic line with $B^{\delta}$, in logarithmic time of $N$.


For an interval $X$ such that $\min X= \frac{i_1}{N}$ and $\max X=\frac{i_2}{N}$ for $i_1,i_2 \in \{0,1,\ldots,N\}$, we write $\med X =  \lceil \frac{i_1 + i_2}{2N} \rceil$.

Consider Algorithm~\ref{alg:FPTAS}. It consists of two stages. 
First, by binary search (Lines \ref{line:1} -- \ref{line:4}), it identifies the vertical rectangle $X\times [0,1]$ of horizontal length $\frac 1N$ where the image of one vertical side by ${\bar f}$ intersects with $L^{\delta}$ and the image of the other side by ${\bar f}$ intersects with $R^{\delta}$. After the first stage of the algorithm, it identifies the vertical edges $e_1$ and $e_2$ where these intersections occur; at this stage, we have identified basic squares $S_1$ and $S_2$ as in Figure~\ref{fig:average}.
The algorithm then by binary search (Lines \ref{line:firstwhile:second} -- \ref{line:11}) updates a vertical rectangle $X\times Y$. Even if the intersection test is done only with at most six edges at Line~\ref{line:if}, Lemma~\ref{lem:secondwhile} ensures that the boundary of that rectangle keeps its image intersecting with both $L^{\delta}$ and $R^{\delta}$. The rectangle therefore ends as a basic square of size $\frac 1N \times \frac 1N$ with the same property. 
By Lemma~\ref{lem:square:LR}, such a square contains a point whose image is $\omega^{\delta}$, and thus a point whose image is $\omega$ by Lemma~\ref{lem:perturb-tri}; this implies by Lemma~\ref{lem:approximate-basicrectangle} that any vertex of the square corresponds to a desired approximate division. The proof of Theorem~\ref{thm:two-layers:group-birth:FPTAS} consists mainly in verifying that Algorithm~\ref{alg:FPTAS} computes correctly such a basic square and in analyzing its running time. 

\begin{algorithm}
\caption{Computing an approximate envy-free multi-division}
\label{alg:FPTAS}
\SetInd{0.8em}{0.3em}
~set $X \coloneqq [0,1]$\;
\While{$\max X -  \min X  > \frac 1N$}{\label{line:1}        
\If(\tcp*[h]{This can be done using binary search; see Lemma~\ref{lem:degree-vertical-line}}){${\bar f}(\bigL(\min X)) \cup {\bar f}(\bigL(\med X))$ intersects with $L^{\delta}$ and with $R^\delta$
\label{line:2}
}{
set $X \coloneqq [\min X,\med X]$\;\label{line:3}
}
\Else{
set $X \coloneqq [\med X, \max X]$\;\label{line:4}
}
}
~compute a basic square $S_1$ in $X \times [0,1]$ which admits a vertical edge $e_1$ whose image by ${\bar f}$ intersects with $L^{\delta}$\;\tcp*[h]{This can be done using binary search; see Lemma~\ref{lem:degree-vertical-line}}\label{line:5}

~compute a basic square $S_2$ in $X \times [0,1]$ which admits a vertical edge $e_2$ whose image by ${\bar f}$ intersects with $R^{\delta}$\;\tcp*[h]{This can be done using binary search; see Lemma~\ref{lem:degree-vertical-line}}\label{line:6}

\If{$S_1=S_2$}{
\Return $S_1$\;
}
\Else{
~set $x_1 \coloneqq \min X$ and $x_2=\max X$\;\label{line:7}
~set $y_1 \coloneqq \min \{\, y \mid (x,y) \in S_1 \cup S_2 \,\}$ and $y_2 \coloneqq \max \{\, y \mid (x,y) \in S_1 \cup S_2 \,\}$\;
~set $Y \coloneqq [y_1,y_2]$\;\label{line:8}
\While{$\max Y -  \min Y  > \frac 1N$}{\label{line:firstwhile:second}
\If{\label{line:if}
$\bigcup_{e}\bar f(e)$ intersects with $L^{\delta}$ and with $R^{\delta}$, where $e$ ranges over the edges incident to the corner vertices of $X \times [\min Y, \med Y]$
\label{line:10}
}{
set $Y \coloneqq [\min Y, \med Y]$\;
}
\Else{
set $Y \coloneqq [\med Y, \max Y]$\;\label{line:11}
}\label{line:endwhile:second}
}
\Return $X \times Y$\;
}
\end{algorithm}



According to Lemma~\ref{lem:usquare:odd}, the intersection of the image of any basic line with $B^{\delta}$ is unique. It can be computed by binary search.

\begin{lemma}\label{lem:degree-vertical-line}
Given any $x=\frac kN$ for $k \in \{0,1,\ldots,N\}$, one can compute in time $O(n\log N)$ the vertical edge of $\bigL(x)$ whose image by ${\bar f}$ intersects with $L^{\delta}$ or $R^{\delta}$. 
\end{lemma}
\begin{proof}
The point $\bar f(x,0)$ lies below $B^{\delta}$ and the point $\bar f(x,1)$ lies above $B^{\delta}$. Since $\bar f$ is continuous, and ${\bar f}(\vv)$ can be computed in $O(n)$ time, binary search allows to compute the vertical edge of $\bigL(x)$ whose image intersects with $B^{\delta}$ in $O(n\log N)$ time.~
\end{proof}


The next lemma ensures that an invariant is kept over the second {\bf while} loop. The proof of this lemma is actually very close to that of the end of Lemma~\ref{lem:average}.

\begin{lemma}\label{lem:secondwhile}
At each iteration of the second {\bf while} loop of Lines \ref{line:firstwhile:second} -- \ref{line:11}, 
the intersections of $\bar f(\partial(X\times Y))$ with $L^{\delta}$ and with $R^{\delta}$ are both nonempty.
\end{lemma}
\begin{proof}
Denote by $F$ the set of edges incident with the corner vertices of $X\times Y$. Notice that except at the end of the last iteration, $F$ contains exactly six edges.
We prove by induction that both $\bar f(F)\cap L^{\delta}$ and $\bar f(F)\cap R^{\delta}$ are nonempty. At the first iteration, this is true because $F$ contains all the vertical edges of $S_1$ and $S_2$. We turn now to the subsequent iterations.

To deal with these iterations, we need to establish that all edges inside $X\times Y$, just before the first iteration, have their image by $\bar f$ intersecting with $B^{\delta}$. Before the first iteration of the while loop, the images of the top and bottom edges in $F$ do not intersect with $B^{\delta}$ since ${\bar f}$ is affine on each edge, the corner vertices of the top edge have their images lying above $B^{\delta}$, and the corner vertices of the bottom edge have their images lying below $B^{\delta}$. The image of each basic line intersects exactly once with $B^{\delta}$ (Lemma~\ref{lem:usquare:odd}) and exactly two edges of each basic square in $X\times Y$ have their image by $\bar f$ intersecting with $B^{\delta}$ (Lemma~\ref{lem:square}). Hence, all the horizontal edges inside $X \times Y$ have their images intersecting with $B^{\delta}$. This property is obviously kept along all iterations.

Now, consider an arbitrary iteration. Denote by $F^{\text{top}}$ the three top edges in $F$, and by $F^{\text{bottom}}$ the three bottom edges in $F$. By induction, $\bar f(F^{\text{top}} )\cup \bar f(F^{\text{bottom}})$ intersects with both $L^{\delta}$ and $R^{\delta}$. Let $e_0$ be the horizontal edge with the second coordinate equal to $\med Y$. Before any iteration (even the last one), $e_0$ is inside $X \times Y$ and thus $\bar f(e_0)$ has a nonempty intersection with $B^{\delta}$. If it misses one of $L^{\delta}$ and $R^{\delta}$, one of $\bar f(F^{\text{top}})$ and $\bar f(F^{\text{bottom}})$ (at least) does not miss it. So, in any case, $\bar f(F^{\text{top}}\cup\{e_0\})$ or $\bar f(F^{\text{bottom}}\cup\{e_0\})$ intersects with both $L^{\delta}$ and $R^{\delta}$.
~
\end{proof}

\smallskip
We are ready to prove Theorem~\ref{thm:two-layers:group-birth:FPTAS}.

\begin{proof}[Proof of Theorem~\ref{thm:two-layers:group-birth:FPTAS}.]
We start by proving the correctness of the algorithm. In the first {\bf while} loop of Lines \ref{line:1} -- \ref{line:4}, Lemmas~\ref{lem:sym} and~\ref{lem:usquare:odd} together ensure that, for the initial interval $X=[0,1]$, the intersections of $\bar f(\bigL(\min X)) \cup \bar f(\bigL(\max X))$ with $L^{\delta}$ and with $R^{\delta}$ are both nonempty. By Lemma~\ref{lem:usquare:odd} applied on $\bar f(\bigL(\med X))$, this property is kept for each iteration.
 This means that one can find $S_1$ and $S_2$ satisfying the conditions in Lines~\ref{line:5} and \ref{line:6}. If $S_1=S_2$, Algorithm~\ref{alg:FPTAS} clearly computes a desired square by Lemma~\ref{lem:square:LR}.
Otherwise, Lemmas~\ref{lem:square:LR} and~\ref{lem:secondwhile} applied at the end of the algorithm ensure that the rectangle $X \times Y$ is a basic square whose image by ${\bar f}$ contains $\omega^{\delta}$ and therefore $\omega$ by Lemma~\ref{lem:perturb-tri}. By Lemma~\ref{lem:approximate-basicrectangle}, any vertex of such a square induces a desired approximate division. 

It remains to analyze the running time of Algorithm~\ref{alg:FPTAS}.
The \textbf{while} loop in Lines~\ref{line:1} -- \ref{line:4} makes $O(\log N)$ iterations and Line~\ref{line:2} can be implemented in time $O(n\log N)$ by Lemma~\ref{lem:degree-vertical-line}. Thus, Lines~\ref{line:1} -- \ref{line:4} can be implemented in $O(n\log^2 N)$. 
Again, by Lemma~\ref{lem:degree-vertical-line}, Lines~\ref{line:5} and \ref{line:6} can be implemented in time $O(n\log N)$. %
Line \ref{line:10} can be implemented in $O(n)$ time, because this can be checked by computing ${\bar f}(\vv)$ of each vertex $\vv$ that belongs to an edge incident to the corner vertices of $X \times [\min Y, \med Y]$. Further, the \textbf{while} loop in Lines~\ref{line:firstwhile:second} -- \ref{line:endwhile:second} makes $O(\log N)$ iterations. 
Thus, the algorithm computes a basic square $X \times Y$ of size $\frac 1N \times \frac 1N$ in $O(n\log^2 N)$ time.
\end{proof}

\subsection{FPTAS for a one-layered cake}
In this section, we provide a proof sketch for Theorem~\ref{thm:one-layer:group-birth:FPTAS}. 
Again, we assume that the $n$-th agent is the birthday agent. 
We keep the same tools as in Section~\ref{subsec:preliminaries}: the unit square, the triangulation $\T$, and the function ${\bar f}$. We also take the same $N$ as in Section~\ref{sec:two:FPTAS}. With monotonicity of preferences, the one-layered case is equivalent to the two-layered problem where the bottom layer is empty, which means that all agents weakly prefer the top layer to the bottom. This will enable us to establish that any point $(a_1,a_2,a_3) \in \Delta^2$ lies in the image of $\bar f$, i.e., $\bar f$ is surjective. This surjectivity will allow eventually, via standard arguments, to show that any choice for $k_1,k_2,k_3$ is achievable. 

\begin{proof}[Proof sketch of Theorem~\ref{thm:one-layer:group-birth:FPTAS}.] Let $\omega=\frac{1}{n-1}(k_1-\frac{1}{3},k_2-\frac{1}{3},k_3-\frac{1}{3}) \in \Delta^2$. We define $B=\{\, (z_1,z_2,z_3) \mid z_1 = \omega_1, (z_1,z_2,z_3) \in \Delta^2\,\}$. Again, $L$ and $R$ are obtained by splitting $B$ by $\omega$. 
Like Lemmas~\ref{lem:perturb-rel} and~\ref{lem:perturb-tri}, we choose a sufficiently small $\delta>0$ to obtain a perturbed version $B^{\delta}$ of $B$ so that no vertex in $\T$ has its image by ${\bar f}$ lying on $B^{\delta}$ and so that if the middle point $\omega^{\delta}$ of $B^{\delta}$ belongs to $\bar f(\tau)$ for some triangle $\tau$ in $\T$, then $\omega$ also belongs to $\bar f(\tau)$. Define $L^{\delta}$ and $R^{\delta}$, analogously as in Section~\ref{subsec:two:existence}, by splitting $B^{\delta}$ by $\omega^{\delta}$.

We have $\bar f_1(x,1)=1$ for every $x=\frac{k}{N}$ with $k \in \{0,1,\ldots,N\}$ since every agent in $[n-1]$ answers $\mathcal{A}_1(x,1)$ as a preferred piece. We have $\bar f_1(x,0)=0$ for every $x=\frac{k}{N}$ with $k \in \{0,1,\ldots,N\}$; indeed, such a piece is never chosen by the fact that $\mathcal{A}_1(x,0)$ is of zero-length and by the tie-breaking rule of Section~\ref{subsec:preliminaries}. (When the valuation functions are monotone and Lipschitz, the hungry assumption is satisfied; see Section~\ref{sec:prelim}.) Similarly, we have $\bar f_2(0,y)=0$ (resp. $\bar f_3(1,y)=0$) for every $y=\frac{k}{N}$ with $k \in \{0,1,\ldots,N\}$. (From this, we could already conclude via standard techniques from topology that $\bar f$ is surjective, but we provide an algorithmic proof in the sequel.)
 
As we have just seen, ${\bar f}(x,1)$ and ${\bar f}(x,0)$ lie above and below $B^{\delta}$ for each $x=\frac{k}{N}$ with $k \in \{0,1,\ldots,N\}$, respectively. This implies, together with the vertical monotonicity of ${\bar f}$ (Lemma~\ref{lem:monotone}), that ${\bar f}(\bigL(x))$ intersects with $B^{\delta}$ at a single point for every $x=\frac{k}{N}$ with $k \in \{0,1,\ldots,N\}$. Thus, we obtain a statement corresponding to Lemma~\ref{lem:usquare:odd}. 

Also the images ${\bar f}(\bigL(0))$ and ${\bar f}(\bigL(1))$ of the sides are symmetric with respect to $z_2=z_3$ (Lemma~\ref{lem:sym}) 
because ${\bar f}(\bigL(0))$ (resp. ${\bar f}(\bigL(1))$) coincides with the boundaries of $\Delta^2$ where the third (resp. the second) coordinate is $0$. It is easy to see that Lemmas~\ref{lem:square} and~\ref{lem:square:LR} remain true for the new $B^{\delta}$. Algorithm~\ref{alg:FPTAS} still applies, Lemmas~\ref{lem:degree-vertical-line} and~\ref{lem:secondwhile} remain true, and the proof of Theorem~\ref{thm:two-layers:group-birth:FPTAS} is still valid. Thus, we are able to find a basic square whose image by $\bar f$ contains $\omega$ in $O(n\log^2 N)$.


Let $(x^*,y^*)$ be a point in this square mapped by $\bar f$ on $\omega$. To see that such a square induces a desired multi-division, define $w_{ij}={\bar f}^{(i)}_{j}(x^*,y^*)$ for every non-birthday agent $i \in [n-1]$ and a layered piece $j \in \{1,2,3\}$. By definition of ${\bar f}$ and the fact that ${\bar f}_j(x^*,y^*) = \omega_j= \frac 1 {n-1}(k_j - \frac 1 3)$ for all $j=1,2,3$, we have 
$$
\sum_{i=1}^{n-1} w_{ij} = \sum_{i=1}^{n-1}{\bar f}^{(i)}_j(x^*,y^*)=(n-1){\bar f}_j(x^*,y^*) = k_j-\frac {1}{3}
$$
for all $j=1,2,3$. Further, $\sum_{j=1}^3 w_{ij}=\sum_{j=1}^3 {\bar f}^{(i)}_j(x^*,y^*) = 1$, for all $i \in [n-1]$ because ${\bar f}^{(i)}$ has its image in $\Delta^2$. 

Now, consider the bipartite graph $H=([n-1],\{1,2,3\};E)$ where the edge $ij$ exists precisely when ${\bar f}^{(i)}_{j}(x^*,y^*)>0$. 
Applying Lemma~\ref{lem:TUM} to $H$ with $a_j=k_j -\frac{1}{3}$ for $j \in \{1,2,3\}$, for every $j^* \in \{1,2,3\}$, there exists an assignment $\pi_{j^*} \colon [n] \rightarrow \{1,2,3\}$ such that 
\begin{itemize}
    \item $\pi_{j^*}(n)=j^*$, 
    \item for each $i \in [n-1]$, the vertex $\pi_{j^*}(i)$ is a neighbor of $i$ in $H$, 
    \item $|\pi^{-1}_{j^*}(n)| = k_{j^*}$, and
    \item for each $j \in \{1,2,3\}$, we have $|\pi^{-1}_{j^*}(j)| \in \{k_j -1, k_j \}$.
\end{itemize}
The last two conditions imply that $|\pi^{-1}_{j^*}(j)| = k_j$ for each $j \in \{1,2,3\}$ since $\sum_{j \in \{1,2,3\}}|\pi^{-1}_{j^*}(j)|=n=\sum_{j \in \{1,2,3\}} k_j$. Then, similar to the proof of Lemma~\ref{lem:approximate-basicrectangle}, for any vertex $\vv \in V(\T)$ that belongs to the basic square containing $(x^*,y^*)$, we have a multi-division $\calA(\vv)$ such that for every $j^* \in \{1,2,3\}$, it admits a desired assignment $\pi_{j^*}$ where   
\begin{itemize}
    \item $\pi_{j^*}(n)=j^*$,
    \item for each $i\in [n-1]$, $v_i(\calA_{\pi_{j^*}(i)}(\vv))+\varepsilon \ge \max_{j \in \{1,2,3\}}v_i(\calA_{j}(\vv))$, and
    \item for each $j\in \{1,2,3\}$, $|\pi^{-1}_{j^*}(j)|=k_j$.~\qedhere
\end{itemize} 
\end{proof}


\section{Proportional multi-division}\label{sec:prop}
In this section, we discuss proportionality. Consider an instance of the multi-layered cake-cutting problem with $m$ layers and $n$ agents, $m\le n$. Let $q$ be an integer such that $m\le q\le n$. A multi-division into $q$ layered pieces is {\em proportional} if there exists a surjective assignment $\pi$ of the agents to the pieces such that each agent $i$ is assigned a layered piece of value at least her {\em proportional fair share}, that is, $\frac{1}{q}$ of $i$'s valuation for the entire layered cake.

A valuation function $v_i$ is {\em additive} if for any pair of layered pieces $\calL$ and $\calL'$ whose interiors are disjoint, we have $v_i(\calL \cup \calL')=v_i(\calL)+v_i(\calL')$, where the union $\calL \cup \calL'$ of two layered pieces $\calL=(L_{\ell})_{\ell \in [m]}$ and $\calL'=(L'_{\ell})_{\ell \in [m]}$ is defined to be $(L_{\ell}  \cup L'_{\ell})_{\ell \in [m]}$. Under additive valuations, it is easy to see that envy-freeness implies proportionality.


\begin{proposition}\label{prop:EFprop}
Consider an instance of the multi-layered cake-cutting problem with $m$ layers and $n$ agents, $m\le n$, with additive valuations. Any envy-free multi-division into $q$ layered pieces with the corresponding assignment is proportional. 
\end{proposition}
\begin{proof}
Take any envy-free multi-division $\mathcal{A}=(\mathcal{A}_1,\mathcal{A}_2,\ldots,\mathcal{A}_q)$ and the corresponding assignment $\pi\colon[n] \rightarrow [q]$. Consider any agent $i$. By envy-freeness, $v_i(\calA_{\pi(i)}) \ge v_i(\calA_{j})$ for all $j \in [q]$. Summing these inequalities, we have $q \cdot v_i(\mathcal{A}_{\pi(i)}) \ge \sum^q_{j=1}v_i(\mathcal{A}_j)$, which by additivity, implies that the valuation of the piece assigned to $i$ is greater or equal to the proportional fair share for $i$. 
\end{proof}

Proposition \ref{prop:EFprop} and Theorem \ref{thm:group-birth} imply that there exists a proportional multi-division into $q$ layered pieces that is feasible and contiguous and satisfies the property that the number of agents assigned to each layered piece differs by at most one when $q$ is a prime power, $m \le q \le n$, and agents have additive valuations. It turns out that such a division exists even when $q$ is not a prime power.

\begin{theorem}\label{thm:proport}
Consider an instance of the multi-layered cake-cutting problem with $m$ layers and $n$ agents, $m\le n$, with additive valuations. Let $q$ be an integer such that $m\le q\le n$. Then there exist a feasible and contiguous multi-division into $q$ layered pieces and an assignment of the agents to the layered pieces so that
\begin{itemize}
    \item each agent is assigned to a layered piece of value at least her proportional fair share.
    \item the number of agents assigned to each layered piece differs by at most one.
\end{itemize}
\end{theorem}
\begin{proof}
Given an instance of the multi-layered cake-cutting problem with $m$ layers and $n$ agents, $m\le n$, recursively repeat the following steps until $q$ reduces to a prime number:
\begin{itemize}[leftmargin=40pt, rightmargin=2pt]
    \item[\emph{Step 1.}] Choose any prime power $q'$ in the prime decomposition of $q$. Construct a new cake $C'$ of $q'$ layers. Each layer $\ell'$ of the new cake is obtained by merging $c_{\ell'}$ consecutive layers of the original instance where $c_{\ell'} \in \{\lfloor m/{q'}\rfloor, \lceil m/{q'}\rceil\}$ for each $\ell' \in [q']$ and $\sum^{q'}_{\ell'=1}c_{\ell'}=m$. Define a new additive valuation of each agent $i \in [n]$ for the new instance. Namely, for a layered piece $\mathcal{L'}=(L'_{\ell'})_{\ell' \in [q']}$ of the new instance, 
    $i$ assigns to each $L'_{\ell'}$ the sum of values of the pieces that correspond to $L'_{\ell'}$ in $(\sum^{\ell'-1}_{k=1}c_{k}+1)$-th, $\ldots$, $(\sum^{\ell'}_{k=1}c_{k})$-th layers of the original instance; and $i$ assigns the sum of such values to $\mathcal{L'}$. 

    \item[\emph{Step 2.}] Apply Theorem~\ref{thm:group-birth} to this cake $C'$ with the whole set of agents, the number of layers being $q'$, and the number of groups being $q'$, and obtain an envy-free multi-division $\mathcal{A}'=(\mathcal{A}'_1,\mathcal{A}'_2,\ldots,\mathcal{A}'_{q'})$ among $q'$ groups of agents and an assignment $\pi' \colon [n] \rightarrow [q']$ satisfying the conditions in Theorem~\ref{thm:group-birth}. See Figure~\ref{fig:prop} for an illustration of the merge and division $\mathcal{A}'$ when $m=5$, $n=13$, and $q=6$.

    \item[\emph{Step 3.}] 
    For each $j' \in [q']$, create a new multi-layered cake $C_{j'}$ by concatenating the allocated layered pieces in $\mathcal{A}'_{j'}=(L'_{\ell'})_{\ell' \in [q']}$ into a multi-layered cake of $\lceil m/q'\rceil$ layers: 
    Each $\ell$-th layer of the multi-layered cake $C_{j'}$ is obtained by concatenating the $\ell$-th layers of the original pieces corresponding to each $L'_{\ell'}$ for $\ell' \in [q']$; 
    in case $c_{\ell'}<\ell$, complete $L'_{\ell'}$ with empty layers with valuation $0$ for all agents so that we can still consider that all layers of $C_{j'}$ are $[0,1]$ intervals. See Figure~\ref{fig:prop} for an illustration of the concatenation. 
    For each $j' \in [q']$, apply induction (i.e., Theorem~\ref{thm:proport}) on $C_{j'}$ with the agent set ${\pi'}^{-1}(j')$ and the number of groups equal to $q/q'$. By collecting the multi-divisions and assignments for all $j' \in [q']$, obtain the final multi-division $\mathcal{A}$ and the final assignment $\pi \colon [n] \rightarrow [q]$ for the original instance. 
\end{itemize}
 
We show that $\mathcal{A}$ and $\pi$ form a desired solution by an induction on the size of the prime decomposition of $q$. If $q$ is a prime power, then by Proposition \ref{prop:EFprop}, this is the case. Suppose that $q=d  q'$ where $q'$ is a prime power in the prime decomposition of $q$. 
Clearly, $\mathcal{A}$ is contiguous and feasible by the induction hypothesis and by the fact that a new multi-layered cake $C_{j'}$ in Step $3$ is obtained by concatenating contiguous pieces of $m$ distinct layers of the original instance. 
Further, since $\mathcal{A}'$ obtained in Step 2 is envy-free, every agent $i$ has value at least $\frac{1}{q'} \alpha_i$ for the assigned layered piece $\mathcal{A}'_{\pi'(i)}$ in $C'$, where $\alpha_i$ denotes $i$'s value for the entire original layered cake (Proposition~\ref{prop:EFprop}); also, by induction hypothesis, she obtains a piece of value at least $\frac{1}{d}$ of $i$'s valuation for $\mathcal{A}'_{\pi'(i)}$. Thus, the piece $\mathcal{A}_{\pi(i)}$ assigned to $i$ has a value at least her proportional fair share. 
Finally, we show that the number of agents assigned to each layered piece $\mathcal{A}_{j}$ is $\lfloor n/q \rfloor$ or $\lceil n/q \rceil$. 
By Theorem~\ref{thm:group-birth}, $|{\pi'}^{-1}(j')| \in \{\lfloor {n}/{q'} \rfloor, \lceil {n}/{q'} \rceil\}$ for each $j' \in [q']$. By the induction hypothesis, the number of agents $|\pi^{-1}(j)|$ assigned to a layered piece $\mathcal{A}_{j}$ is of the form 
$\lfloor {|{\pi'}^{-1}(j')|}/{d} \rfloor$ or $\lceil {|{\pi'}^{-1}(j')|}/{d} \rceil\}$ for some $j'\in [q']$. Thus, we have
\[
\Bigl \lfloor \frac 1d \Bigl \lfloor \frac n {q'} \Bigr \rfloor \Bigr \rfloor \le |\pi^{-1}(j)| \le \Bigl \lceil \frac 1d \Bigl \lceil \frac {n}{q'} \Bigr \rceil \Bigr  \rceil \, .
\]
Now, since $q=d q'$ and since
\[
\Bigl \lfloor \frac 1c \Bigl \lfloor \frac ab \Bigr \rfloor \Bigr \rfloor = \Bigl \lfloor \frac{a}{bc} \Bigr\rfloor
~~\mbox{and}~~
\Bigl \lceil \frac 1c \Bigl \lceil \frac ab \Bigr \rceil \Bigr  \rceil = \Bigl \lceil \frac{a}{bc} \Bigr \rceil
\]
hold for all integers $a,b,c$, we have $|{\pi}^{-1}(j)| \in \{\lfloor {n}/{q} \rfloor, \lceil {n}/{q} \rceil\}$ for each $j \in [q]$.
\end{proof}


\begin{figure}[hbt]
\centering
\begin{tikzpicture}[scale=0.5, transform shape]
        \begin{scope}[xshift=-14cm,yshift=0.5cm]
            \draw[thick] (2,0) rectangle (8,2);
            \draw[dotted,thick] (2,1) -- (8,1); 
            
            \draw[thick] (2,-2) rectangle (8,0);
            \draw[dotted,thick] (2,-1) -- (8,-1); 
            
            \draw[thick] (2,-3) rectangle (8,-2); 
            \draw[ultra thick,->] (8.5,-0.5) -- (9.5,-0.5);

            \node[draw] at (5,-4) {\Large 1. Merge layers};
        \end{scope}   

        \begin{scope}[xshift=-6cm,yshift=0.5cm]
            \draw[thick] (2,0) rectangle (4,2);
            \draw[thick,fill=red!30] (4,0) rectangle (6,2);
            \draw[thick,fill=blue!30] (6,0) rectangle (8,2);
            \draw[dotted,thick] (2,1) -- (8,1); 
            
            \draw[thick,fill=red!30] (2,-2) rectangle (4,0);
            \draw[thick,fill=blue!30] (4,-2) rectangle (6,0);
            \draw[thick] (6,-2) rectangle (8,0);
            \draw[dotted,thick] (2,-1) -- (8,-1); 
            
            \draw[thick,fill=blue!30] (2,-3) rectangle (4,-2);
            \draw[thick] (4,-3) rectangle (6,-2);
            \draw[thick,fill=red!30] (6,-3) rectangle (8,-2);  

            \node[draw] at (5,-4) {\Large 2. Obtain an envy-free division};
        \end{scope} 

        \begin{scope}[xshift=2cm,yshift=0.5cm]
             \draw[thick,fill=red!30] (4,0) rectangle (6,2);
            \draw[dotted,thick] (4,1) -- (6,1); 
            
            \draw[thick,fill=red!30] (2,-2) rectangle (4,0);
            \draw[dotted,thick] (2,-1) -- (4,-1); 
        
            \draw[thick,fill=red!30] (6,-3) rectangle (8,-2); 
            \draw[ultra thick,->] (8,-0.5) -- (9,-0.5);

            \node at (5,1) {\Large $L'_1$};
            \node at (3,-1) {\Large $L'_2$};
            \node at (7,-2.5) {\Large $L'_3$};
            
            \node[draw] at (9,-4) {\Large 3. Concatenate each allocated piece and apply induction};
        \end{scope}
         \begin{scope}[xshift=10cm,yshift=0.5cm]

            \path[fill=red!30] (2,-1) rectangle (4,0);
            \path[fill=red!30] (4,-1) rectangle (5,0);
            \path[fill=red!30] (5,-2) rectangle (6,-1);

            \draw[thick] (2,-2) rectangle (4,0);
            \draw[thick] (4,-2) rectangle (6,0);
            \draw[thick] (6,-1) rectangle (8,0); 
            \draw[dotted,thick] (2,-1) -- (4,-1); 
            \draw[dotted,thick] (4,-1) -- (6,-1);

            \draw[thick,dotted] (5,0.5) -- (5,-2.5);
        \end{scope}
\end{tikzpicture}
\caption{Illustration of the proof of Theorem~\ref{thm:proport} when $m=5$, $n=13$, and $q=6~(=2 \cdot 3)$. Here, we take $q'=3$.}
\label{fig:prop}
\end{figure}
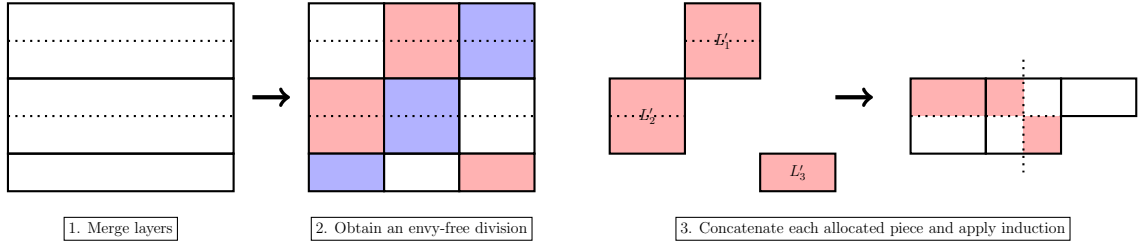

\section{Concluding remarks}\label{sec:concluding}

Theorem~\ref{thm:group-birth} states the existence of an envy-free multi-division of a multi-layered cake using $q-1$ long knives for $n$ agents when $q$ is a prime power. The proof relies crucially on Volovikov's theorem, a powerful result from equivariant topology. We discuss below a further implication of our result and some limitations of our approach.


\paragraph{Envy-free cake division with pieces of equal size.} Theorem~\ref{thm:group-birth} has the following interesting consequence to the standard cake-cutting problem. It implies in particular that when $q$ is a prime power there always exists an envy-free division with pieces of equal size, meaning that each agent receives a piece of an equal total length. This may be relevant when agents' bundles need to fulfill some restrictions, e.g., the employees of companies may have maximum weekly working hours.

In the statement, we use the expression ``group-birthday'' to refer to the properties stated in Theorem~\ref{thm:group-birth} for the assignment.

\begin{proposition}
Consider an instance of the one-layered cake-cutting problem with $n$ agents, with closed preferences. Let $q$ be a nonnegative integer nongreater than $n$. If $q$ is a prime power, then there exists an envy-free division into $q$ pieces with a ``group-birthday'' assignment such that
\begin{itemize}
    \item each piece is formed by $q$ subintervals.
    \item the lengths of the subintervals are the same for each piece.
\end{itemize}
\end{proposition}
\begin{proof}
We divide the one-layered cake into $q$ intervals of length $1/q$, and arrange them into $q$ layers to create an instance of a multi-layered cake-cutting. By a straightforward analysis of the proof of Theorem~\ref{thm:group-birth}, we see that the envy-free multi-division satisfies the stated properties when converted to the original single-layered cake-cutting problem.
~
\end{proof}

 We note that our result shares a similar flavor as Theorem 6.14 of Joji{\'c} et al. \cite{jojic2019splitting}, which ensures the existence of an envy-free division such that each agent gets a piece of the same length under various measures.


\paragraph{Limitation of the approach based on Volovikov's theorem.}
For a number $q$ not equal to a prime power, it is known that Volovikov's theorem does not hold. Further, for such a $q$, as already noted in the introduction,
there is an example of a one-layered cake-cutting instance with non-hungry choice functions for which no envy-free division among $q$ agents exists~\cite{Avvakumov_2020,panina2021envyfree}. These examples show some limitations of the approach based on Volovikov's theorem, but do not prohibit the existence of an envy-free division in the valuation function model. 
Indeed, Avvakumov and Karasev \cite{avvakumov2020equipartition} showed that an envy-free division among any number of agents still exists when agents have identical valuations that are not necessarily hungry; it does not seem that this result can be obtained from Volovikov's theorem. See also the last paragraph of Section 1 in \cite{avvakumov2020equipartition}. 
In the context of multi-layered cake-cutting, the examples of \cite{Avvakumov_2020,panina2021envyfree} imply that an envy-free multi-division using $q-1$ long knives may not exist under the choice function model. However, it is still an open question whether or not a counterexample exists for some natural valuation functions. 

\paragraph{Computational complexity.}
While we prove the existence of an envy-free multi-division using $q-1$ long knives, we do not settle the precise complexity class to which the problem belongs. On a related note, the computational problem of the BSS theorem, which is a special case of the Volovikov's theorem, has been recently shown to be PPA-$p$-complete~\cite{FHSZ_soda2021}. 
Due to the relation that PPA-$p$ $=$ PPA-$p^k$ (see Proposition $2.2$ of~\cite{FHSZ_soda2021}), it would be quite surprising if our problem belongs to PPA-$p^k$, which would then imply that our existential result can be proven via the BSS theorem, a less powerful statement than that by Volovikov. Hence, we expect that a new complexity class, encompassing the Volovikov theorem, should probably be introduced, though the challenge lies in the fact that there is no constructive proof of Volovikov's theorem. 

\section*{Acknowledgments}
This work was partially supported by JST PRESTO Grant Numbers JPMJPR20C1. 

\bibliographystyle{plain}


\newpage
\appendix
\section{Some basic notions from combinatorial topology}\label{sec:basics}
\subsection{Abstract and geometric simplicial complexes} Given a finite set $V$, a collection $\K$ of subsets of $V$ is an {\em abstract simplicial complex} if whenever $\sigma\in\K$ and $\tau\subseteq \sigma$, then $\tau\in\K$. The elements of $\K$ are the {\em simplices} of the abstract simplicial complex. The elements of $V$ occurring in at least one simplex of $\K$ are the {\em vertices} of $\K$ and their set is denoted by $V(\K)$.

A finite collection $\Gamma$ of geometric simplices (convex hulls of affinely independent points) is a {\em geometric simplicial complex} if the following two conditions are satisfied:
\begin{itemize}
    \item the faces of every $\sigma$ in $\Gamma$ are all in $\Gamma$.
    \item for every pair $\sigma,\sigma'$ of simplices in $\Gamma$, the intersection of $\sigma$ and $\sigma'$ is a face of both $\sigma$ and $\sigma'$.
\end{itemize}
The empty set is considered as a face of any simplex. The {\em underlying space} of a geometric simplicial complex $\Gamma$ is the set of all points contained in at least one simplex of $\Gamma$. It is denoted by $\|\Gamma\|$ and we have thus $\|\Gamma\| = \bigcup_{\sigma\in \Gamma}\sigma$.


There is a strong relation between abstract and geometric simplicial complexes. The vertex sets of the simplices of a geometric simplicial complex form an abstract simplicial complex. The former is then a {\em geometric realization} of the latter. On the other hand, any abstract simplicial complex admits a geometric realization by taking generic representatives of its vertices in a sufficiently large Euclidean space. The equivalence between abstract and geometric simplicial complexes makes that one often switches from one point of view to the other without further mention, and the same notation $\K$ may be used for an abstract simplicial complex and its geometric realization.

\subsection{Simplicial maps and triangulations} Given two abstract simplicial complexes $\K_1$ and $\K_2$, a map $\lambda\colon V(\K_1)\rightarrow V(\K_2)$ is a {\em simplicial map} of $\K_1$ into $\K_2$ if $\lambda(\sigma)\in\K_2$ for every $\sigma\in\K_1$. A bijective simplicial map $\lambda$ of $\K_1$ into $\K_2$ is an {\em isomorphism} if its inverse map is also a simplicial map. It is an {\em automorphism} if moreover $\K_1=\K_2$. 

Any such simplicial map $\lambda$ induces a natural continuous map $\|\Gamma_1\|\rightarrow\|\Gamma_2\|$, where $\Gamma_1$ and $\Gamma_2$ are geometric realizations of respectively $\K_1$ and $\K_2$, as we explain now. The {\em affine extension} of $\lambda$, denoted 
by ${\bar \lambda}$, is the map $\|\Gamma_1\|\rightarrow\|\Gamma_2\|$ obtained by extending the original $\lambda$ affinely to the relative interiors of the simplices of $\Gamma_1$: given a simplex $\sigma$ in $\Gamma_1$ with vertices $\vv_1,\ldots,\vv_k$, a point $\x=\sum_{i=1}^k\alpha_i\vv_i$ in the relative interior of $\sigma$ with $\alpha_1,\ldots\alpha_k\ge 0$, and $\sum_{i=1}^k\alpha_i=1$, we define ${\bar \lambda}(\x)$ as the point $\sum_{i=1}^k\alpha_i\lambda(\vv_i)$. It can be proved that for every such simplicial map $\lambda$, the map ${\bar \lambda}$ is continuous. If $\lambda$ is an isomorphism, then ${\bar \lambda}$ is a homeomorphism.

A {\em triangulation} of a topological space $X$ is a simplicial complex (abstract or geometric) whose underlying space is homeomorphic to $X$. By a natural extension, a triangulation of a simplicial complex is defined as a triangulation of its underlying space.

\subsection{Group acting on simplicial complexes and topological spaces}

Let $(G,\cdot)$ be a finite group, whose neutral element is denoted by $e$.

An {\em action} of $G$ on an abstract simplicial complex $\K$ is a collection $(\varphi_g)_{g\in G}$ of simplicial automorphisms $\varphi_g$ of $\K$ such that
\begin{itemize}
    \item[$\bullet$]  $\varphi_e = \id_\K$ (where $\id_\K$ is the simplicial identity map of $\K$), and
    \item[$\bullet$]  $\varphi_g\circ\varphi_h = \varphi_{g\cdot h}$ for all $g,h\in G$.
\end{itemize}

Similarly, an {\em action} of $G$ on a topological space $X$ is a collection $(\psi_g)_{g\in G}$ of homeomorphisms $\psi_g\colon X\rightarrow X$ such that
\begin{itemize}
    \item[$\bullet$] $\psi_{e} = \id_X$, and
    \item[$\bullet$]  $\psi_g\circ\psi_h = \psi_{g\cdot h}$ for all $g,h\in G$.
\end{itemize}
 For a point $x\in X$, the set $\{\psi_g(x)\colon g \in G\}$ is the {\em orbit} of the point $x$ under the $G$-action.
 
 Note that if $(\varphi_g)_{g\in G}$ is an action of $G$ on an abstract simplicial complex $\K$ with geometric realization $\Gamma$, then $({\bar \varphi}_g)_{g\in G}$ is an action of $G$ on $\|\Gamma\|$. This latter action is considered without further mention when the action has only been introduced for an abstract simplicial complex and when the underlying space of a geometric realization is then considered.
 
Sometimes, when there is no ambiguity, $\varphi_g(v)$ and $\varphi_g(x)$ are simply denoted by $g\cdot v$ and $g\cdot x$.

\subsection{Connectivity}

A topological space $X$ is $d$-connected if every continuous map ${\bar f}\colon S^k\rightarrow X$ with $k\in\{-1,0,\ldots,d\}$ can be extended to a continuous map $\bar f\colon B^{k+1}\rightarrow X$. Here, $S^k$ denotes the $k$-dimensional sphere and $B^k$ the $k$-dimensional ball. In this context, $S^{-1}$ is interpreted as $\varnothing$ and $B^0$ as a single point. Therefore, $(-1)$-connected means nonempty.

\end{document}